\documentclass[pra,twocolumn,superscriptaddress]{revtex4-2}
\usepackage{amsthm}
\usepackage{amsmath}
\usepackage{latexsym}
\usepackage{amsfonts}
\usepackage{amssymb}
\usepackage{color}
\usepackage{bbm,dsfont}
\usepackage{graphicx}
\usepackage{xcolor}
\definecolor{linkblue}{RGB}{0,70,140}
\usepackage[colorlinks=true,
            linkcolor=linkblue,
            citecolor=linkblue,
            urlcolor=linkblue]{hyperref}
\usepackage{subfigure}
\usepackage{wasysym}
\usepackage{mathrsfs}

\definecolor{cobalt}{rgb}{0.0, 0.28, 0.67}

\usepackage{tikz}
\usetikzlibrary{arrows.meta}
 
\newtheorem{proposition}{Proposition}
\newtheorem{proposition?}{Proposition?}
\newtheorem{theorem}{Theorem}
\newtheorem{lemma}{Lemma}
\newtheorem{corollary}{Corollary}

\theoremstyle{definition}

\newcommand{\ip}[2]{\left\langle\,#1\,|\,#2\,\right\rangle} 
\newcommand{\ket}[1]{|#1\rangle} 
\newcommand{\kb}[2]{|#1\rangle\langle#2|} 
\newcommand{\tr}[1]{\textrm{tr}\left[#1\right]} 
\newcommand{\rank}{\mathrm{rank}\,} 

\newcommand{\id}{\mathbbm{1}} 

\usepackage{bm}
\DeclareMathOperator{\vol}{vol}

\newcommand{\F}{\mathsf{F}}
\newcommand{\G}{\mathsf{G}}
\newcommand{\Q}{\mathsf{Q}}
\renewcommand{\P}{\mathsf{P}}
\newcommand{\M}{\mathsf{M}}
\newcommand{\N}{\mathsf{N}}

\newcommand{\MICs}{\mathcal P}
\newcommand{\ext}{\mathcal E}
\newcommand{\Eobs}{\mathcal E_{\geq}(r)}
\newcommand{\Ep}{\mathcal E_{>}(r)}

\newcommand{\Mone}{\mathscr M_1(r)}
\newcommand{\MIC}{\mathscr M_{\rm MIC}(r)}

\newcommand{\itlabel}[1]{\textnormal{(}#1\textnormal{)}}

\usepackage{soul}

\begin{document}
\title[]{Characterising extremal decoherence by quantum measurement incompatibility}

\author{Daniel McNulty}
\affiliation{Dipartimento di Fisica, Università di Bari, 70126 Bari, Italy}
\affiliation{INFN, Sezione di Bari, 70126 Bari, Italy}
\author{William Townsend}
\author{Jukka Kiukas}
\affiliation{Department of Mathematics, Aberystwyth University, Aberystwyth, SY23 3BZ, United Kingdom}

\date{July 23, 2026}

\begin{abstract} 
Decoherence, when viewed in the Heisenberg picture, can turn incompatible measurements into jointly measurable ones. This provides an observable-level description of emergent classicality in terms of the incompatibility destroyed by a noise channel. The corresponding loss of incompatibility can be captured by the channel’s \emph{compatibility region}---the observables in a chosen probe class that become jointly measurable with every noisy observable. The geometry and volume of this region provide a refined operational way to compare noise channels. We apply this framework to extremal decoherence channels, each of which acts by Schur multiplication with extreme points of the convex set of unit-diagonal positive semidefinite matrices. Restricting to extremals eliminates convex mixing and reveals subtle phase-dependent decoherence effects that are not captured by damping rates or conventional coherence quantifiers. Such channels have a rigid dilation structure described by a rank-one operator frame, which reduces joint measurability to a positivity test, making the compatibility region analytically tractable and its volume computable from the frame’s Gram matrix. If the frame can be normalised to a minimal informationally complete (MIC) POVM, the compatibility region acquires a geometric interpretation in the associated probability representation of the quantum state space, familiar from QBism. Among maximal-rank extremals, those associated with symmetric informationally complete (SIC) POVMs maximise the compatibility volume and hence destroy the greatest amount of incompatibility, providing an operational characterisation of the special role of SICs in terms of joint measurability. This leads to a new reformulation of the well-known SIC existence problem as a joint measurability question for noisy mutually unbiased observables.
\end{abstract}

\maketitle

\section{Introduction}

Decoherence---the mechanism by which quantum systems acquire classical features through environment interaction \cite{breuer02,schlosshauer07,schlosshauer14}---can be described not only in terms of quantum states, but also at the level of observables. In the Heisenberg picture, one relevant notion of classicality is \emph{joint measurability} \cite{uola14}: a family of observables is jointly measurable if their outcome statistics can be recovered from the marginals of a single parent observable. Conversely, lack of joint measurability---\emph{measurement incompatibility}---is a signature of non-classicality; it is considered a resource for creating quantum correlations \cite{uola15,quintino14,kiukas17}, and a special form of quantum contextuality \cite{abramsky11,xu19}. From this perspective, the onset of joint measurability under decoherence provides an observable-level notion of emergent classicality.

This raises the possibility of characterising quantum noise through the measurement incompatibility it destroys or preserves. This viewpoint fits within the broader resource-theoretic study of quantum dynamics \cite{liu19,takagi19,liu20,li20}, in which channels are compared through their ability to generate, transform, or preserve nonclassical resources \cite{theurer19,hsieh20,saxena20,takahashi22,ku22,stratton24}; this has also recently been applied to incompatibility \cite{hsieh25} by employing incompatibility breaking channels (originally introduced in \cite{heinosaari15}) as the free channels. With this approach, however, the channel's incompatibility preservation would be quantified by a scalar robustness monotone which does not reveal which observables retain or lose incompatibility, nor how extensively the ``preserved'' incompatibility is distributed across the measurement space.

To retain this finer structure, some of the present authors developed a set-valued framework for characterising incompatibility loss under decoherence without dissipation \cite{kiukas22}. Relative to a preferred basis, such channels preserve diagonal elements and damp off-diagonal ones; in finite dimension they are described by Schur multiplication with a coherence matrix, i.e., a positive semidefinite matrix with unit diagonal \cite{buscemi05}. They are also commonly referred to as \emph{pure dephasing} or \emph{phase-damping} channels, and as \emph{genuinely incoherent operations} in the resource theory of coherence \cite{devicente16,streltsov17}. Throughout this work, we refer to this class simply as \emph{decoherence channels}.

For these channels, incompatibility loss is probed using incoherent observables, namely those diagonal in the preferred basis. Although mutually compatible and invariant under decoherence, nontrivial incoherent observables may be incompatible with other measurements. The channel's \emph{compatibility region} is defined as the set of such observables that become jointly measurable with every decohered observable. Its complement therefore identifies those observables whose incompatibility survives, at least with one noisy measurement. Under further decoherence, the compatibility region expands and hence defines a set-valued noise monotone, retaining information about both the extent and the structure of incompatibility loss. Along divisible open system dynamics, this expansion describes the progressive loss of incompatibility in the quantum-to-classical transition \cite{kiukas22,kiukas23}.

The present work extends this framework by focusing on the extremal points of the convex set of decoherence channels. These channels are particularly interesting because they isolate the genuinely quantum part of decoherence, after excluding the additional classical randomness arising from convex mixing. Apart from the unitary extremals, for which no damping occurs, extremal decoherence channels are non-random-unitary. Their action therefore cannot be represented as classical fluctuations, modelled by a random mixture of diagonal unitary evolutions, and therefore requires a genuinely quantum description of the environment \cite{gregoratti03}. Such channels have been studied through the mathematical structure of their associated correlation matrices \cite{grone90,li94,kiukas08}, as well as through physical realisations \cite{helm09,helm11,trendelkamp11,pernice12,kayser15}.

From a technical point of view, the extremal case is far more constrained than the general one, and therefore calls for a different toolbox. Our main results build on the observation that extremal decoherence channels have a rigid structure naturally described by rank-one operator frames. In particular, the rank-one projections arising from a Gram representation of the coherence matrix span the full operator space on the dilation Hilbert space. As a result, compatibility reduces from a semidefinite feasibility problem to a direct positivity test, allowing analytic and geometric descriptions of the compatibility region. For maximal-rank extremal channels, we find that the Euclidean volume of the compatibility region is determined by the Gram matrix of the associated rank-one operator frame.

One of our main tools turns out to be the theory of \emph{minimal informationally complete} POVMs (MICs), which has been studied extensively in foundations, especially in the context of Quantum Bayesianism \cite{debrota19,debrota20,debrota21}, and in applications such as tomography \cite{bent15,garcia21,acharya21,gupta25,chen15,bae19,carmeli11,carmeli12}. MICs give rise to extremal decoherence channels, and this correspondence allows us to describe the compatibility region geometrically, as a membership problem in the associated MIC state space \cite{debrota20,appleby11}. The case in which the MIC is a symmetric informationally complete measurement (SIC) is especially revealing. The compatibility criterion then reduces to positivity of the SIC state-reconstruction formula \cite{fuchs13,appleby11b,debrota20}, and joint measurability can be expressed directly in terms of the SIC form of the Born rule. In QBism, this relation is viewed as a quantum modification of the classical law of total probability \cite{fuchs11,fuchs13,appleby11b,debrota20}. The SIC case therefore provides a particularly symmetric instance of the general connection between extremal decoherence, informational completeness, compatibility, and the probability geometry of quantum states.

SIC-extremals also have distinctive quantitative properties. Within the class of maximal-rank extremals, they maximise the compatibility volume and therefore destroy the greatest amount of incompatibility. This ordering does not persist, however, when extremal channels are compared with \emph{non}-extremal channels that have the same damping rates. By restricting the probe class to the one-parameter family of a depolarised basis observable, we show that a SIC-extremal destroys less incompatibility than the corresponding phase-insensitive uniform channel with the same off-diagonal damping magnitudes. In fact, asymptotically, as the system size increases, the SIC-extremal preserves all incompatibility within this probe family, whereas uniform decoherence destroys it. This contrast shows that the phase relations encoded by the coherence matrix, which are not captured by the damping rates alone, can have a decisive effect on incompatibility loss. 

Our results also yield an operational reformulation of the SIC existence problem \cite{zauner99,fuchs17}: a SIC exists in $\mathbb C^r$ if and only if there exists a maximal-rank extremal decoherence channel on $\mathbb C^{r^2}$ that renders a specific noisy pair of mutually unbiased observables jointly measurable. In this sense, Zauner's conjecture can be recast entirely as a question about incompatibility loss under decoherence.

Finally, we illustrate the general theory in dimension $d=4$ by giving analytic descriptions of the compatibility regions, together with closed-form expressions for their volumes, for several families of extremal channels. For MIC-extremals, these regions admit a simple geometric description in terms of ellipsoids, with the SIC case distinguished as the unique spherical instance. Our examples include a two-parameter family of extremals arising from qubit Heisenberg--Weyl MICs \cite{ariano04}, a one-parameter family arising from qubit semi-SICs \cite{geng21}, and a one-parameter family not associated with MICs \cite{buscemi05}.

The paper is organised as follows. We begin in Sec.~\ref{sec:sic-preview} with a SIC example, illustrating the main compatibility criteria that will be generalised in later sections. Sec.~\ref{sec:background} introduces the general framework for characterising noise via emergent joint measurability, its application to decoherence channels, and some background on MIC state spaces. Section \ref{sec:ext-dec} discusses extremal decoherence channels and their relation to rank-one operator frames and MICs. In Sec.~\ref{sec:jm-ext} we derive general and MIC-based characterisations of joint measurability under extremal decoherence. In Sec.~\ref{sec:jm-sic} we analyse the SIC case and highlight its distinctive features. In Sec.~\ref{sec:qubit-mics} we illustrate the results for several families of extremal decoherence channels when $d=4$. Finally, we provide a brief summary and outlook in Sec.~\ref{sec:conclusion}.

\section{A preview: the SIC case} \label{sec:sic-preview}

\begin{figure*}[t]
    \centering
    \includegraphics[width=0.8\textwidth]{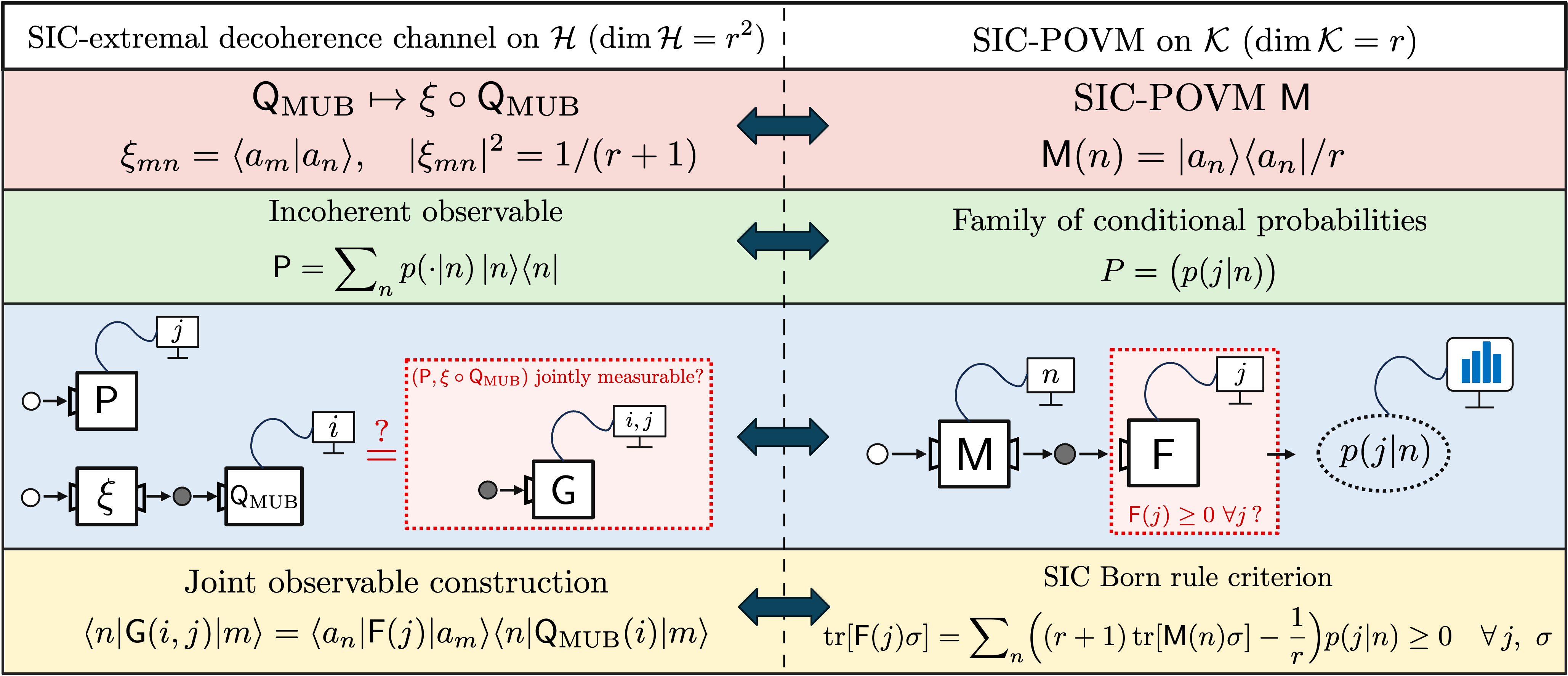}
    \caption{
Correspondence between symmetric extremal decoherence and SIC-POVMs. A SIC-POVM $\M$ determines the coherence matrix $\xi$, while an incoherent observable $\P$ is specified by the conditional probabilities $p(j|n)$. The pair $(\P,\xi\circ\Q_{\rm MUB})$ is jointly measurable iff these probabilities arise from a valid POVM $\F$ in the SIC representation, i.e. when the associated SIC Born rule expression is non-negative for every state.
}
    \label{fig:preview}
\end{figure*}

Before developing the general theory, we describe the main idea in the particularly transparent case of a SIC. Fix an incoherent basis $\{|n\rangle\}_{n=1}^{d}$ of $\mathcal H\simeq\mathbb C^d$, and let $\P=\sum_n p(\cdot|n)|n\rangle\langle n|$ be an incoherent observable. Here the conditional probabilities $p(j|n)$ define a column-stochastic matrix, with columns indexed by the incoherent basis labels $n$ and rows by the outcomes $j$. A decoherence channel acts in the Heisenberg picture as $A\mapsto \xi\circ A$, where $\xi$ is a positive semidefinite matrix with unit diagonal and $\circ$ denotes Schur multiplication. Since $\xi$ has unit diagonal, incoherent observables are left unchanged by decoherence, hence $\xi\circ\P=\P$. We are interested in the compatibility region associated with $\xi$: the set of incoherent observables that become jointly measurable with $\xi\circ\Q$, for every observable $\Q$. As shown in \cite{kiukas22}, it is enough to test this condition against a single maximally coherent observable. In particular, $\P$ belongs to the compatibility region if and only if it is jointly measurable with $\xi\circ\Q_{\rm MUB}$, where $\Q_{\rm MUB}$ is a rank-one projective observable in a basis mutually unbiased to the incoherent one \cite{durt10,mcnulty26}. Thus, for decoherence, the full compatibility region is determined by joint measurability with this single decohered mutually unbiased observable.

Suppose that $\M$ is a SIC-POVM on $\mathcal K\simeq\mathbb C^r$, with outcomes labelled by the incoherent basis labels $n$. In particular, we have $d=r^2$ and we write $\M(n) = r^{-1}|a_n\rangle\langle a_n|$ where $|a_n\rangle$ are unit vectors in $\mathcal K$. From the SIC we obtain a coherence matrix $\xi$ by setting $\xi_{nm}=\langle a_n|a_m\rangle$. This matrix has unit diagonal, satisfies $|\xi_{nm}|^2=1/(r+1)$ for all $n\neq m$, and is extremal in the convex set of coherence matrices. We denote by $\MICs$ the set of probability distributions obtained from measuring $\M$ on quantum states $\sigma$ on $\mathcal K$, i.e., vectors $\mathbf s\in\mathbb R^d$ with components $s(n) = {\rm tr}[\M(n)\sigma]$. Since $\M$ is informationally complete, the state $\sigma$ corresponding to each $\mathbf s\in\MICs$ is unique. In the literature, $\MICs$ is called the \emph{SIC state space} since it identifies quantum states with a subset of the classical probability simplex $\Delta_d$ \cite{fuchs13,appleby11b,debrota20}. Our general results, when applied to this example of high symmetry, yield the following explicit characterisation of joint measurability. 
\begin{theorem}\label{motivationthm}
Suppose $\M(n)=r^{-1}|a_n\rangle\langle a_n|$, $n=1,\ldots,r^2$, is a SIC-POVM on $\mathcal K\simeq\mathbb C^r$, and write $\xi_{nm}=\langle a_n|a_m\rangle$ for the associated coherence matrix. Then, for any incoherent observable $\P(j)=\sum_n p(j|n)|n\rangle\langle n|$ on $\mathcal H\simeq\mathbb C^{r^2}$, the following are equivalent:
\begin{enumerate}
\item[\itlabel{i}] $\P$ and $\xi\circ \Q$ are jointly measurable for every observable $\Q$ on $\mathcal H$;
\item[\itlabel{ii}] $\P$ and $\xi\circ \Q_{\rm MUB}$ are jointly measurable;
\item[\itlabel{iii}] There exist matrices $\xi(j)\geq 0$ such that $\sum_j \xi(j)=\xi$ and $\xi_{nn}(j)=p(j|n)$, for all $j,n$;
\item[\itlabel{iv}] $\sum_{n}\left[(1+r)\frac{p(j|n)}{{\rm tr}[\P(j)]} - \frac{1}{r}\right]\M(n)\geq 0$ for all $j$;
\item[\itlabel{v}] $\sum_{n}\!\left[(r+1)s(n)-\frac{1}{r}\right]\!p(j|n)\geq 0$ for all $j$ and $\mathbf s\in \mathcal P$;
\item[\itlabel{vi}] $\frac{1}{{\rm tr}[\P(j)]}\bigl(p(j|1),\dots,p(j|r^2)\bigr)^T\in \mathcal P$ for all $j$.
\end{enumerate}
\end{theorem}
Condition (i) is the joint measurability statement we wish to characterise, and its equivalence to (ii) and (iii) is a nontrivial result from \cite{kiukas22} which hold for arbitrary, not necessarily extremal, coherence matrices $\xi$. In particular, condition (iii) is the general semidefinite matrix completion criterion for joint measurability under decoherence (cf. Theorem~\ref{thm:gii}). Condition (iv) gives the main simplification obtained in the present paper under extremality. Instead of searching for a feasible matrix completion, one forms the operator displayed in (iv) directly from $p(j|n)$ and the SIC $\M$, reducing compatibility to a positivity check for each outcome $j$ with no optimisation required. Conditions (v) and (vi) are equivalent reformulations in terms of the SIC state space, with the latter showing that membership in the compatibility region corresponds to the normalised rows of the column-stochastic matrix $P=(p(j|n))$ lying inside the SIC state space. The rest of the paper develops this theorem in full generality. Theorem \ref{thm:extr} gives the positivity criterion for maximal-rank extremals, and Theorem~\ref{mainthm} gives its MIC state space form. These characterisations will allow us to compare extremal decoherence channels via their compatibility regions, and to prove that SIC-extremals maximise the volume among them.

Remarkably, condition (v) corresponds exactly to the QBism formulation of the Born rule, also called the \emph{quantum law of total probability}, for a SIC reference measurement \cite{fuchs13,appleby11b,debrota20}.  To see this, let $\sigma$ be the state corresponding to the SIC probability vector $\mathbf s=(s(1),\ldots,s(d))^T\in\mathcal P$, with $s(n)=\tr{\M(n)\sigma}$. If we denote the expression in condition (v) by $p(j)$ and the operator in condition (iv) by $\F(j)$, then, whenever the equivalent conditions of the theorem hold, $j\mapsto\F(j)$ is a POVM and $j\mapsto p(j)$ is a probability distribution. Taking the trace of condition (iv) against $\sigma$ gives $p(j)=\tr{\F(j)\sigma}$. Thus $p(j)$ is the probability of outcome $j$ in a direct measurement of $\F$ on $\sigma$. On the other hand, $p(j|n)=\tr{\F(j)\Pi_n}$, with $\Pi_n=|a_n\rangle\langle a_n|$, is the probability of outcome $j$ when $\F$ is measured after obtaining SIC outcome $n$, whose L\"uders post-measurement state is $\Pi_n$. As illustrated in Fig.~\ref{fig:preview}, this demonstrates that $\P$ is compatible with $\xi\circ \Q_{\rm MUB}$ if and only if the values $p(j|n)$ can be interpreted as the conditional probabilities of a genuine measurement $\F$ evaluated on the SIC states $\Pi_n$. Condition (v) is therefore exactly the SIC/QBist relation between the probabilities for a direct measurement and those for a measurement performed after the reference SIC.

In the QBism literature, this relation is compared with the classical law of total probability, which would read $p(j)=\sum_n p(j|n)s(n)$. In the SIC representation, the Born rule replaces the classical weights $s(n)$ by the expression $(r+1)s(n)-1/r$, and this correction distinguishes the quantum relation from the classical one. In our present framework, the classical law is recovered from the incoherent observable $\P$: if the SIC probability vector $\mathbf s$ is embedded into the diagonal sector as $\rho=\sum_n s(n)|n\rangle\langle n|$, then $\tr{\P(j)\rho}=\sum_n p(j|n)s(n)$. Thus the ordinary law of total probability governs the statistics of the classical observable $\P$ on diagonal states, whereas condition (v) gives the SIC Born rule statistics of the observable $\F$.

From the SIC perspective, the equivalence with condition~(ii) brings out a \emph{complementarity} aspect arising from the coherence structure relative to the diagonal sector. The incoherent observable $\P$ is a classical post-processing of the incoherent basis measurement. Therefore, whenever $\P$ is nontrivial, the pair $(\P,\Q_{\rm MUB})$ retains a complementarity relation: $\P$ carries information about the incoherent/SIC label, while $\Q_{\rm MUB}$ is maximally coherent with respect to that label. In Pauli's terminology \footnote{The term is from Pauli's well-known 1926 letter to Heisenberg.}, $\P$ probes the system through the ``p-eye'', whereas $\Q_{\rm MUB}$ probes the complementary ``q-eye''. In the SIC representation, the same complementarity is reflected in the replacement of the classical weights $s(n)$ by the Born rule weights $(r+1)s(n)-1/r$. This replacement is the probabilistic signature of the fact that the SIC reference measurement and the subsequent measurement cannot be treated as ordinary jointly measurable classical random variables. In the decoherence picture, the same obstruction appears as measurement incompatibility: although $\P$ is an ``unsharp'' \cite{busch89} version of the incoherent basis observable, it remains incompatible with $\Q_{\rm MUB}$ unless it is trivial. Decoherence then ``blurs the q-eye'', replacing $\Q_{\rm MUB}$ by $\xi\circ\Q_{\rm MUB}$, and the theorem determines when this additional imprecision is sufficient to make the observables jointly measurable. We stress that this complementarity aspect of decoherence-induced loss of incompatibility is present for arbitrary coherence matrices, as shown in \cite{kiukas22}. What is special in the extremal case is the emergence of the direct positivity conditions~(iii)--(iv), together with the link to informationally complete measurements developed below.

\section{Background}\label{sec:background}

\subsection{Observables and channels in the context of decoherence}

We first introduce the relevant notation for observables and channels used throughout. Let $\mathcal H$ be a Hilbert space of $\dim \mathcal H=d<\infty$ and fix an orthonormal basis $\{|n\rangle\}_{n=1}^d$ of $\mathcal H$, called the \emph{incoherent basis}. All notions of coherence are understood relative to this basis.

We let $\mathcal S(\mathcal H)=\{\sigma\in M_d(\mathbb C) \mid \sigma\geq 0,\, {\rm tr}[\sigma]=1\}$ denote the \emph{state space} associated to $\mathcal H$. Here $M_d(\mathbb C)$ is the set of $d\times d$ complex matrices, and we often write $A\geq 0$ to mean that $A\in M_d(\mathbb C)$ is positive semidefinite.

An \emph{observable} (POVM) $\M$ on $\mathcal H$, with a finite outcome set $\Omega_\M$, consists of positive semidefinite matrices (or effects) $\M(i)\in M_d(\mathbb C)$ such that $\sum_{i\in \Omega_\M} \M(i)=\id$. 

An observable $\P$ is called \emph{incoherent} if each of its effects is diagonal in the incoherent basis, i.e.
$
\P(j)=\sum_{n=1}^d p(j|n)\,|n\rangle\langle n|,
$
where, for each $n$, the map $j\mapsto p(j|n)$ is a probability distribution on some outcome set $\Omega_\P$. In other words, incoherent observables are the classical post-processings of the sharp measurement in the incoherent basis. We occasionally let $P=(p(j|n))_{j,n}$ denote the column-stochastic matrix corresponding to $\P$. We denote by $\mathscr{I}_{\!d}$ the set of all incoherent observables on $\mathcal H$. We note that $\mathscr I_{\!d}$ contains the set of \emph{trivial} observables $\mathscr O_{\rm triv}$ of the form $\P(j) =p(j)\id$.

Let $\Lambda$ be a Heisenberg picture quantum channel, i.e. a completely positive unital map acting on $d\times d$ matrices. (All channels in the paper are considered in the Heisenberg picture unless otherwise stated.) In the Schrodinger picture, the pre-dual channel $\Lambda_*$ acts on the state space $\mathcal S(\mathcal H)$, and is trace-preserving. We say that $\Lambda$ describes \emph{decoherence} relative to the incoherent basis if $\Lambda$ leaves every diagonal operator unchanged; equivalently, $\Lambda(|n\rangle\langle n|)=|n\rangle\langle n|$ for every basis projection. Such channels have the form
\begin{equation}\label{eq:schur-channel}
\Lambda_\xi(A)=\xi\circ A,
\end{equation}
where $\circ$ denotes the entrywise (Schur) product in the incoherent basis and $\xi$ is a positive semidefinite matrix with unit diagonal \cite{buscemi05,helm09}. We call $\xi$ the \emph{coherence matrix} of the channel, and let
\begin{equation}\label{eq:gios}
\mathfrak C_d
:=
\{\xi\in M_d(\mathbb C)\mid \xi\ge0,\ \xi_{nn}=1\ \text{for all }n\}\,,
\end{equation}
denote the set of such matrices. Their role in the context of measurement incompatibility was studied in \cite{kiukas22}.

\begin{figure*}[t!]
    \centering
\includegraphics[width=0.80\textwidth]{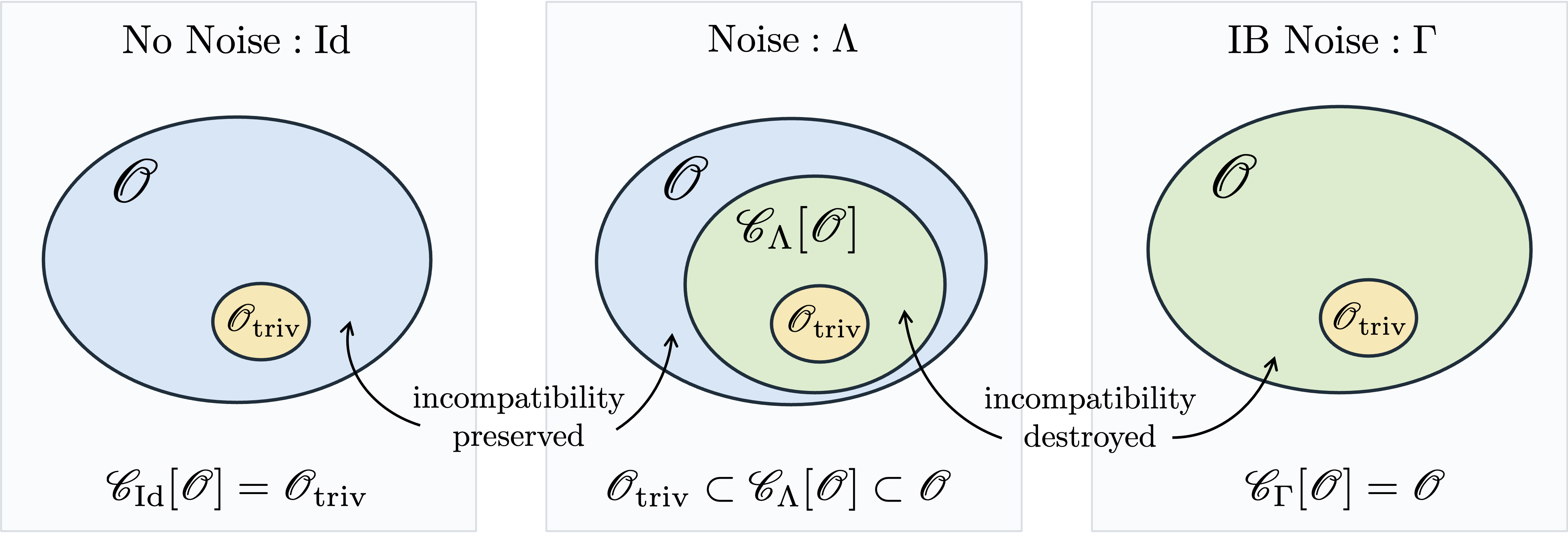}
    \caption{
Characterising noise through the compatibility region
$\mathscr C_\Lambda[\mathscr O]$ of a fixed ``probe'' class of observables $\mathscr O$. This region consists of the probes that become jointly measurable with every noisy observable under the channel $\Lambda$. With no noise, $\mathscr C_{\rm Id}[\mathscr O]=\mathscr O_{\rm triv}$; for a general noise channel, $\mathscr C_\Lambda[\mathscr O]$ lies between $\mathscr O_{\rm triv}$ and $\mathscr O$; and for incompatibility-breaking noise $\Gamma$, $\mathscr C_\Gamma[\mathscr O]=\mathscr O$. In this work, we focus on decoherence channels $\Lambda_\xi$, for which the natural probe class is the set of incoherent observables, while the incompatibility-breaking reference channel $\Gamma$ is complete dephasing.}
    \label{fig:compatibility-region-schematic}
\end{figure*}

Any $\xi\in \mathfrak C_d$ admits (due to positivity) a Gram decomposition
\begin{equation}\label{eq:coh_matrix}
\xi_{nm}=\langle a_n|a_m\rangle,
\end{equation}
for suitable unit vectors $|a_n\rangle$ in an auxiliary Hilbert space $\mathcal K$. Taking $\mathcal K=\mathrm{span}\{|a_n\rangle\,:\,n=1,\dots,d\}$, one has
\begin{equation}
r:=\dim\mathcal K=\rank \xi.
\end{equation}
This representation provides the minimal Stinespring dilation of the decoherence channel via the isometry
\begin{equation}\label{eq:isometry}
V:\mathcal H\to\mathcal H\otimes\mathcal K,
\qquad
V|n\rangle=|n\rangle\otimes |a_n\rangle,
\end{equation}
as
\[
\Lambda_\xi(A)=V^*(A\otimes\id_{\mathcal K})V=\xi\circ A\,.
\]
The vectors $|a_n\rangle$, which we call the \emph{structure vectors} of the channel, therefore describe its action in the Heisenberg picture: the $(n,m)$ matrix element of an observable is multiplied by the overlap $\langle a_n|a_m\rangle$. The minimal family of structure vectors is unique up to a common unitary transformation on $\mathcal K$.

\subsection{Characterising quantum noise via emergent joint measurability}\label{subsec:generalcompreg}

Here we take a step back to formalise the central concept of this paper: a \emph{compatibility region associated with a quantum channel}. We do this at a general level before specialising to decoherence, thereby separating the general idea from the particular application studied below.

Let $\mathscr M$ denote the set of all observables on a fixed Hilbert space $\mathcal H$. If $\Lambda$ is a quantum channel and $\M\in \mathscr M$, we let $\Lambda(\M)$ denote the ``noisy'' observable $\Lambda(\M)(i):=\Lambda(\M(i))$, corresponding to measuring $\M$ after the system state has passed through the channel. For a subset $\mathscr O\subseteq \mathscr M$ we write $\Lambda(\mathscr O):=\{\Lambda(\M)\mid \M\in \mathscr O\}$. 

We first recall that $\M,\N\in \mathscr M$ are \emph{jointly measurable} if there is a \emph{joint observable} $\G=(\G(i,j))_{(i,j)}\in \mathscr M$ with $\sum_j \G(i,j)=\M(i)$ for all $i$, and $\sum_i \G(i,j)=\N(j)$ for all $j$; otherwise $\M$ and $\N$ are \emph{incompatible} \cite{busch16,guhne23}. Equivalently, a measurement of $\G$, followed by classical postprocessing, reproduces the outcome statistics of both $\M$ and $\N$ \cite{ali09}. One often considers joint measurability of multiple observables, but the pairwise notion suffices here.

The basic observation is that compatibility of a given pair $\M,\N\in \mathscr M$ is always preserved by a channel $\Lambda$: if $\M$ and $\N$ are compatible, then so are $\Lambda(\M)$ and $\Lambda(\N)$. Incompatibility, by contrast, may be \emph{broken} by $\Lambda$, depending on how noisy the channel is \cite{heinosaari15,kiukas23}.

Let $\mathscr L$ be a semigroup of channels describing the noise type under consideration.  (Note that $\mathscr L$ need not be one-parameter or even commutative.) We assume that $\mathscr L$ contains the identity channel ${\rm Id}$ (no noise), and a fixed \emph{incompatibility-breaking} (IB) channel $\Gamma$. By this we mean that $\Gamma$ maps every finite measurement assemblage to a jointly measurable one. In particular, any $\M,\N\in\Gamma(\mathscr M)$ are compatible. Only this pairwise consequence will be needed in the definition of the compatibility region below. Now $\mathscr L$ carries a natural noise (partial) order: we write $\Lambda\geq\Lambda'$ if there exists $\Phi\in\mathscr L$ such that $\Lambda=\Phi\circ\Lambda'$. Operationally, $\Phi$ represents additional noise applied to the system state \emph{before} it is passed through $\Lambda'$. If $\Lambda'$ breaks the incompatibility of some pair of observables (with a joint observable $\G'$), then so does $\Lambda$ (with $\G = \Phi(\G')$). Thus, once incompatibility is lost, it remains lost further along the noise order. The channel $\Gamma$ represents a
reference noise sufficient to make all observables compatible; in particular, any $\Lambda\geq\Gamma$ is also IB.

While $\Gamma$ destroys \emph{all} incompatibility, a general channel $\Lambda\in\mathscr L$ may destroy only some of it, and the distinction tells us something about the incompatibility breaking capability of the noise in $\Lambda$. To make this precise we choose a ``probe set'' $\mathscr O\subseteq \Gamma(\mathscr M)$ of observables containing the trivial observables $\mathscr O_{\rm triv}$. For a given $\P\in \mathscr O$, the noisy version $\Lambda(\P)$ is then  automatically compatible with $\Lambda(\Q)$ for all $\Q\in \mathscr O$, but \emph{not necessarily} for all $\Q\in \mathscr M$, so this captures the above-mentioned distinction. Accordingly, we define the \emph{compatibility region of $\Lambda$} as
$$
\mathscr C_\Lambda[\mathscr O] \!:=\! \{\P\in \mathscr{O}\,|\,\Lambda(\P) \,\,\text{compatible with}\,\, \Lambda( \Q)\,\forall\, \Q\in \mathscr M\}.
$$
Clearly, $\mathscr O_{\rm triv}\subseteq \mathscr C_\Lambda[\mathscr O]\subseteq \mathscr O$. Thus, $\mathscr C_\Lambda[\mathscr O]$ belongs to the interval $[\mathscr O_{\rm triv},\mathscr O]:= \{\mathscr C\in 2^{\mathscr M}\mid \mathscr O_{\rm triv}\subseteq \mathscr C\subseteq \mathscr O\}$, which we order by set inclusion. We have obtained a natural noise monotone:

\begin{proposition}
The map $$\mathscr L\ni\Lambda\longmapsto \mathscr C_\Lambda[\mathscr O]\in[\mathscr O_{\rm triv},\mathscr O]$$ is increasing in the noise order, with boundary cases $\mathscr C_{\rm Id}[\mathscr O]=\mathscr O_{\rm triv}$ (no noise), and $\mathscr C_\Lambda[\mathscr O] = \mathscr O$ for all $\Lambda\geq \Gamma$ (IB noise).
\end{proposition}

Therefore, the ``extent'' of the set $\mathscr C_\Lambda[\mathscr O]$ characterises the effect of the noise in $\Lambda$ in terms of the compatibility it creates. This idea is illustrated in Fig.~\ref{fig:compatibility-region-schematic}. Thus, we do not have a single number acting as a noise quantifier, but an object of much richer structure. However, we can easily obtain a natural numerical quantifier, namely the (Euclidean) volume:
$$\mathscr L\ni \Lambda \longmapsto \vol (\mathscr C_\Lambda[\mathscr O])\in [0,\infty).$$ This quantity is clearly increasing in the noise order, and we call it the \emph{compatibility volume} of $\Lambda$.

To obtain a dimensionless comparison, one can normalise it with respect to a reference channel $\Lambda_{\rm ref}$. In particular, we define the \emph{compatibility volume} of $\Lambda$ relative to $\Lambda_{\rm ref}$ as
\begin{equation}\label{eq:rel-vol}
\vol(\Lambda;\Lambda_{\rm ref},\mathscr O)
:=
\frac{\vol(\mathscr C_\Lambda[\mathscr O])}
{\vol(\mathscr C_{\Lambda_{\rm ref}}[\mathscr O])}.
\end{equation}
A natural reference is the IB channel $\Gamma$, in which case we denote simply 
\begin{equation}\label{eq:rel-vol_IB}
\vol(\Lambda;\mathscr O):=\vol(\Lambda;\Gamma,\mathscr O)=\vol(\mathscr C_\Lambda[\mathscr O])/\vol(\mathscr O).
\end{equation}
For comparisons within a restricted class of channels, however, a different reference may be more suitable (cf. Sec. \ref{sec:jm-sic}).

Of course, the compatibility region $\mathscr C_\Lambda[\mathscr O]$ will crucially depend on the choice of $\mathscr O$, which can be adjusted to capture the essential features of the specific noise type $\mathscr L$. In the present paper, the relevant noise class is that of decoherence.

Before proceeding to describe this case, we briefly discuss the resource-theoretic aspect. Indeed, the complement of the compatibility region provides a description of how incompatibility is preserved under noise. Namely,
\[
\mathscr P_\Lambda[\mathscr O]
:=
\mathscr O\setminus\mathscr C_\Lambda[\mathscr O]\,,
\]
contains probes $\P$ for which there exists an observable $\Q$ such that $\Lambda(\P)$ and $\Lambda(\Q)$ remain incompatible. This connects our construction with the resource theory of \emph{incompatibility preservability} given in \cite{hsieh25}. In that framework, a channel is resourceful whenever it preserves the incompatibility of at least one measurement assemblage, i.e., whenever it is not incompatibility breaking. This property is quantified by a robustness measure given by the minimum amount of mixing with an auxiliary channel required to make it incompatibility breaking.

The two approaches retain different levels of information. The robustness measures a channel's overall separation from the set of incompatibility-breaking channels, but does not reveal which incompatible observables survive or how extensively it is preserved across the measurement space. By contrast, $\mathscr P_\Lambda[\mathscr O]$ identifies the probes whose incompatibility survives, and describes their geometry within the chosen probe class. The relative volume records the proportion of that class for which incompatibility is preserved. Although probe-dependent, these quantities can be more readily computable and more sensitive to the structure of a particular noise model.

We compare our compatibility region approach briefly with robustness in the concrete context of decoherence channels in Sec.~\ref{subsec:sic_max} and Appendix~\ref{app:resource-preservability}.

\subsection{Compatibility regions for decoherence}

The compatibility region for decoherence was introduced in \cite{kiukas22,kiukas23}. We can now place it within the general framework above and review the aspects of its characterisation from \cite{kiukas22} that are relevant for the extremal case studied in this work.

Here the semigroup $\mathscr L$ consists of all decoherence channels of the form \eqref{eq:schur-channel}, with $\xi\in\mathfrak C_d$. Channel composition corresponds to Schur multiplication of the associated coherence matrices. The incompatibility-breaking reference channel is the complete dephasing channel $\Gamma[A]=\sum_n \langle n|A|n\rangle |n\rangle\langle n|$, corresponding to the coherence matrix $\xi=\mathbb I$. This channel is entanglement-breaking, and hence incompatibility-breaking \cite{heinosaari15}. Clearly, $\Gamma[\mathscr M]=\mathscr I_{\!d}$, and therefore in the decoherence setting the natural probe sets $\mathscr O$ are subsets of $\mathscr I_{\!d}$; these all commute among themselves, which directly shows the IB property. 

For each $\xi\in \mathfrak C_d$ we denote the associated compatibility regions by $\mathscr C_\xi[\mathscr O]:=\mathscr C_{\Lambda_\xi}[\mathscr O]$, and use $\mathscr C_\xi := \mathscr C_\xi[\mathscr I_{\!d}]$. It then follows that
\begin{align}\label{eq:c_xi}
\mathscr{C}_{\xi}[\mathscr O] =&  \{\P\in \mathscr{O}\,|\,\P \,\,\text{is compatible with}\,\, \xi\circ \Q,\,\forall\, \Q\in \mathscr M\}\, \nonumber\\
=& \{\P\in \mathscr{O}\,|\,\P \,\,\text{is compatible with}\,\, \xi\circ \Q_{\rm MUB}\}\,,
\end{align}
where $\Q_{\rm MUB}$ is the projective measurement in a basis mutually unbiased to the incoherent basis. The first equality is immediate from the general definition using  $\xi \circ \P=\P$ for every incoherent observable $\P$. The second equality is a nontrivial result from \cite{kiukas22}: for decoherence, it is enough to test compatibility against a single maximally coherent observable.

Finally, unitary channels do not affect joint measurability and therefore preserve incompatibility, so we would like to eliminate this freedom. It is easy to see that $\xi\in \mathfrak C_d$ defines a unitary channel iff it has rank one, that is, $\xi_{nm}=\overline{v_n} v_m$ for phases $v_n\in \mathbb C$, $|v_n|=1$. Accordingly, we call $\xi, \xi'\in \mathfrak C_d$ \emph{equivalent}, written $\xi\sim\xi'$, if there exists a rank one $\xi_0\in \mathfrak C_d$ such that $\xi = \xi_0\circ \xi'$. By unitary invariance, $\mathscr C_\xi[\mathscr O]$ depends only on the equivalence class of $\xi$.

The following theorem from \cite{kiukas22} gives the basic characterisation of $\mathscr C_\xi$; it is a general version of the equivalence (i) $\Leftrightarrow$ (iii) in Theorem \ref{motivationthm}.

\begin{theorem}\label{thm:gii}
Let $\xi\in \mathfrak C_d$ and $\P=\sum_n p(\cdot|n)|n\rangle\langle n|\in \mathscr I_{\!d}$. The following are equivalent:
\begin{enumerate}
\item[\itlabel{i}] $\P\in\mathscr C_\xi$;
\item[\itlabel{ii}] There exist positive semidefinite matrices $\xi(j)\geq 0$ such that $\sum_j \xi(j)=\xi$ and $\xi_{nn}(j)=p(j|n)$ for all $j,n$.
\end{enumerate}
In this case there exists a POVM $\F$ on the minimal dilation space $\mathcal K={\rm span}\{|a_n\rangle\}_n$ such that
\begin{equation}\label{eq:GIIdilation}
\xi_{nm}(j)=\langle a_n|\F(j)a_m\rangle,
\end{equation}
where $|a_n\rangle$ are structure vectors of $\xi$ given by \eqref{eq:coh_matrix}. Conversely, every POVM $\F$ on $\mathcal K$ defines positive semidefinite matrices through \eqref{eq:GIIdilation} that sum to $\xi$, and hence an incoherent observable $\P^{\F}\in\mathscr C_\xi$ with $p^{\F}(j|n)=\langle a_n|\F(j)|a_n\rangle$.
\end{theorem}
Thus $\P\in\mathscr C_\xi$ if and only if the diagonal
entries $p(j|n)$ can be completed to positive semidefinite matrices $\xi(j)$ summing to $\xi$. As shown in \cite{kiukas22}, the latter condition can be interpreted operationally as saying that $\P$ is measured by a ``genuinely incoherent instrument'' with coherence matrix $\xi$. In general this completion, or equivalently the corresponding POVM $\F$ on the dilation space, is neither fixed by $\xi$ nor unique for a
given $\P$.

From a technical point of view, the semidefinite matrix completion problem in condition~(ii) is typically simpler than the general semidefinite programming (SDP) formulation of joint measurability. In particular, whether $\P$ is a member of $\mathscr C_\xi$ depends only on the relation between the incoherent observable $\P$ and the coherence matrix $\xi$, without explicit reference to any auxiliary observable $\Q$. 

The POVM $\F$ in Theorem~\ref{thm:gii} has a simple interpretation in the dilation picture. It is a measurement on the dilation space $\mathcal K$ appearing in the isometry \eqref{eq:isometry}. In particular, using $\xi_{nm}(j)=\langle a_n|\F(j)a_m\rangle$ and setting $m=n$, the diagonal condition in Theorem~\ref{thm:gii} gives
\[
p(j|n)=\xi_{nn}(j)=\langle a_n|\F(j)|a_n\rangle .
\]
Thus, if the input is the incoherent basis state $|n\rangle$, so that the auxiliary system is in the state $|a_n\rangle$, measuring $\F$ produces the same conditional probabilities as the incoherent observable $\P(j)=\sum_n p(j|n)|n\rangle\langle n|$. Whenever such a POVM $\F$ exists, a joint measurement of $\P$ and $\xi\circ\Q$ is obtained by applying the isometry $V$ and then measuring $\Q$ on $\mathcal H$ and $\F$ on $\mathcal K$. In particular, we have the joint measurement
\begin{equation}\label{eq:jm}
\G(i,j)=V^*(\Q(i)\otimes\F(j))V\,,
\end{equation}
with marginals
\[
\sum_i\G(i,j)=\P(j),
\qquad
\sum_j\G(i,j)=\xi\circ\Q(i)\,.
\]

\subsection{MIC state spaces}\label{sec:mic-state-spaces}

On the dilation space $\mathcal K\simeq\mathbb C^r$, we will make use of \emph{informationally complete} (IC) POVMs. A POVM $\M$ on $\mathcal K$ is informationally complete if its effects span the full operator space, i.e. $\mathrm{span}\{\M(j)\mid j\in\Omega_{\M}\}=\mathcal B(\mathcal K)$. In other words, the outcome probabilities of $\M$ determine the measured state uniquely \cite{busch91}. Since $\dim\mathcal B(\mathcal K)=r^2$, any IC-POVM must have at least $r^2$ outcomes. We will be particularly interested in \emph{minimal informationally complete} POVMs (MICs), namely IC-POVMs that attain this lower bound and therefore have exactly $r^2$ outcomes \cite{debrota20,debrota21}. A special subclass is formed by the \emph{symmetric informationally complete} POVMs (SICs), which are rank-one MICs of the form $\M(n)=|a_n\rangle\langle a_n|/r$, where the unit vectors $|a_n\rangle$ satisfy $|\langle a_n|a_m\rangle|^2=1/(r+1)$, for all $n\neq m$. SICs have been extensively studied \cite{renes04,fuchs17,appleby18,appleby22} and are conjectured to exist in every dimension \cite{zauner99}.

We now briefly review the ingredients concerning MIC-POVMs that will be used below. Let $\Delta_d$ denote the set of probability vectors $\mathbf p=(p_1,\ldots,p_d)$, and let $\M$ be a MIC on $\mathcal K$, with $d=r^2$ outcomes. Measuring $\M$ on a state $\sigma\in\mathcal S(\mathcal K)$ produces the probability vector
\begin{equation}\label{eq:MIC-prob-map}
\bigl[p_\M(\sigma)\bigr]_n
:=
\tr{\M(n)\sigma},
\qquad n=1,\ldots,d.
\end{equation}
The set of all probability distributions obtainable in this way is
\[
\mathcal P(\M)
:=
p_\M\bigl(\mathcal S(\mathcal K)\bigr)
\subseteq\Delta_d,
\]
and is called the \emph{MIC state space} associated with $\M$ \cite{debrota20,appleby11}. It represents the quantum state space as a convex subset of the classical probability simplex, relative to the chosen reference measurement $\M$.

Because $\M$ is informationally complete, the map $p_\M$ is injective and therefore establishes a one-to-one correspondence between $\mathcal S(\mathcal K)$ and $\mathcal P(\M)$. State reconstruction is described by the unique Hilbert--Schmidt dual basis $\{\M^\#(n)\}_{n=1}^d$, defined by $\tr{\M^\#(n)\M(m)}=\delta_{nm}$. Every operator $A\in\mathcal B(\mathcal K)$ can then be reconstructed as $A=\sum_{n=1}^d\tr{\M(n)A}\,\M^\#(n)$. Accordingly, to any probability vector $\mathbf p\in\Delta_d$ we associate the reconstructed operator
\begin{equation}\label{eq:rhoM}
\rho_\M(\mathbf p)
:=
\sum_{n=1}^d p_n\M^\#(n).
\end{equation}
This is the unique operator mapped to $\mathbf p$ by the linear extension of the MIC probability map, since $\tr{\M(m)\rho_\M(\mathbf p)}=p_m$ for every $m$. In particular, $\rho_\M\bigl(p_\M(\sigma)\bigr)=\sigma$ for every state $\sigma$, so $\rho_\M$ is the inverse of $p_\M$ on the MIC state space.

For every $\mathbf p\in\Delta_d$, the reconstructed operator $\rho_\M(\mathbf p)$ is Hermitian and has unit trace. Indeed, $\tr{\M^\#(n)}=\sum_m\tr{\M^\#(n)\M(m)}=1$. It need not, however, be positive semidefinite: although every quantum state determines a probability vector in $\Delta_d$, not every classical probability vector represents a quantum state. Consequently,
\[
\mathbf p\in\mathcal P(\M)
\quad\Longleftrightarrow\quad
\rho_\M(\mathbf p)\geq0.
\]
The MIC state space is therefore the positivity domain of the reconstruction map: it consists of those classical probability vectors whose MIC reconstruction yields a valid quantum state.

For a SIC, the reconstruction map takes an especially simple form. If $\M(n)=\Pi_n/r$, then the dual basis is
\begin{equation}\label{eq:dual_sic}
\M^\#(n)=(r+1)\Pi_n-\id.
\end{equation}
Using $\sum_n\Pi_n=r\id$, Eq.~\eqref{eq:rhoM} becomes
\begin{equation}\label{eq:SIC-reconstruction2}
\rho_\M(\mathbf p)
=
\sum_{n=1}^d
\left((r+1)p_n-\frac1r\right)\Pi_n.
\end{equation}
Thus, a quantum state is represented uniquely by its probabilities for the fixed SIC reference measurement, and the above expression recovers the corresponding density operator. This is the SIC state-reconstruction formula familiar from the QBist literature \cite{fuchs13,renes04,appleby11}.

\section{Extremal decoherence}\label{sec:ext-dec}

The set $\mathfrak C_d$ of coherence matrices is compact and convex. In this section we consider its extremal points, namely the coherence matrices that cannot be written as nontrivial convex combinations of other elements of $\mathfrak C_d$.

Unlike in the classical case, extremal channels need not be unitary \cite{landau93}. This distinction is important since unitary extremals preserve both coherence and incompatibility, whereas non-unitary extremals can degrade both. Apart from the unitary case, which corresponds to $\rank\xi=1$, extremal decoherence channels are non-random-unitary and therefore capture the genuinely quantum part of decoherence. Such non-unitary extremals exist for every $d\geq 4$, and there are extremal coherence matrices of every rank $1<r\leq\sqrt d$ \cite{christensen79,loewy80,li94}.

From the perspective of compatibility regions, we have the following basic observation:

\begin{proposition}\label{convexinclusion}
Let $\xi_k\in \mathfrak C_d$ and let $\xi=\sum_k p_k \xi_k$ be their convex combination, where $p_k\ge 0$ and $\sum_k p_k=1$. Then $\cap_k \mathscr C_{\xi_k} \subseteq \mathscr C_\xi$.
\end{proposition}
\begin{proof}
If $\P\in \cap_k \mathscr C_{\xi_k}$, there exist, for each $k$, a joint observable $\G_k$ of $\P$ and $\xi_k\circ \Q$. It follows immediately that $\G=\sum_k p_k \G_k$ is a joint observable for $\P$ and $\xi\circ \Q$.
\end{proof}

Since every decoherence channel is a convex combination of extremal ones, Proposition~\ref{convexinclusion} shows that the compatibility region of a general channel contains the intersection of the compatibility regions of some extremal components. The additional compatible observables arise from convex mixing, i.e. from extra classical noise. In this sense, extremal decoherence channels capture the essential quantum contribution to the interplay between decoherence and joint measurability. As we will see, this allows us to single out the SIC-case discussed in the preview section as the ``closest to classical'' quantum extremal in terms of the associated compatibility region.

From the technical point of view, extremal coherence matrices admit a rigid description in terms of rank-one operator frames.  This structure makes the joint measurability problem more tractable analytically and leads to surprising links with other fundamental quantum properties, as already discussed in Sec.~\ref{sec:sic-preview}.

\subsection{Extremals of maximal rank}

We begin with the known characterisation of extremal coherence matrices \cite{landau93}, which underlies the results that follow. Let $\xi\in\mathfrak C_d$ have the Gram representation given in Eq.~\eqref{eq:coh_matrix}, with $r=\rank\xi=\dim\mathcal K$. Then $\xi$ is extremal in $\mathfrak C_d$ if and only if
\begin{equation}\label{extremality}
{\rm span}\{|a_n\rangle\langle a_n| \mid n=1,\ldots,d\}
=
\mathcal B(\mathcal K).
\end{equation}
In particular, extremality implies $r^2\leq d$.

Equivalently, an extremal coherence matrix can be viewed as a rank-one \emph{operator frame} \cite{casazza00,ariano04}. By this we mean the family of rank-one projections
\[
\Pi_n^\xi:=|a_n\rangle\langle a_n|,
\qquad n=1,\ldots,d,
\]
whose linear span is $\mathcal B(\mathcal K)$; in fact, the extremality condition
\eqref{extremality} says that $\{\Pi_n^\xi\}_{n=1}^d$ forms such a frame. Conversely, any rank-one operator frame defines an extremal coherence matrix by choosing unit vector representatives $|a_n\rangle$ of the projections $\Pi_n^\xi$ and taking $\xi_{nm}=\langle a_n|a_m\rangle$. Changing the phases of the vectors changes $\xi$ only within its equivalence class and leaves the projections $\Pi_n^\xi$ unchanged. Moreover, different minimal Gram decompositions of the same $\xi$ are related by a unitary on $\mathcal K$, and therefore give unitarily equivalent operator frames. Hence equivalence classes of extremal coherence matrices of rank $r$ are naturally identified with rank-one operator frames on an $r$-dimensional dilation space, up to unitary conjugation.

In this work we focus on the maximal-rank case $d=r^2$, which saturates the extremality bound $r^2\leq d$. The associated operator frame then has no linear redundancy and hence forms a basis of $\mathcal B(\mathcal K)$, giving rise to connections with MICs and SICs developed below. This is also the regime in which the compatibility region is full-dimensional in the natural space of incoherent probes. Accordingly, we denote by $\ext(r)$ the set of equivalence classes of rank-$r$ extremal coherence matrices:
\[
\ext(r)
:=
\{[\xi]\mid \xi\in\mathfrak C_{r^2},\ \rank\xi=r,\ \xi\text{ is extremal in }\mathfrak C_{r^2}\}.
\]
In this case the rank-one operator frame has exactly as many elements as the dimension of
$\mathcal B(\mathcal K)$, and extremality is therefore equivalent to linear
independence. This condition is encoded by the real Gram matrix
\begin{equation}\label{eq:Gxi}
G_\xi:=\xi\circ\bar\xi,
\qquad
(G_\xi)_{nm}=\tr{\Pi_n^\xi\Pi_m^\xi}\,,
\end{equation}
whose entries are the absolute squares of the elements of the coherence matrix $\xi$. Clearly, $G_\xi$ only depends on the equivalence class $[\xi]$. For notational simplicity, we subsequently write $\xi$ instead of $[\xi]$, while keeping in mind that equivalent coherence matrices are considered equal. The significance of $G_\xi$ stems from the following elementary observation:
\begin{proposition}\label{Dinv}
Let $\xi\in \mathfrak C_{r^2}$ be of rank $r$. Then:
\begin{itemize}
\item[\itlabel{a}] $G_\xi$ is positive semidefinite.
\item[\itlabel{b}] $G_\xi$ is positive definite if and only if $\xi\in \ext(r)$.
\end{itemize}
\end{proposition}
\begin{proof}
Choose structure vectors $|a_n\rangle$ such that $\xi_{nm}=\langle a_n|a_m\rangle$ and $\Pi_n^\xi=|a_n\rangle\langle a_n|$. Then
\[
\tr{(\Pi_n^\xi)^*\Pi_m^\xi}
=
\langle a_n|a_m\rangle\langle a_m|a_n\rangle
=
(G_\xi)_{nm},
\]
therefore $G_\xi$ is the Gram matrix of the set $\{\Pi_n^\xi\}_{n=1}^d$ with respect to the Hilbert--Schmidt inner product. In particular, it is always positive semidefinite. Furthermore, since $\dim\mathcal B(\mathcal K)=r^2=d$, the spanning condition \eqref{extremality} is equivalent to linear independence of the $d$ operators $\Pi_n^\xi$, which in turn is equivalent to invertibility of their Gram matrix $G_\xi$. Hence $\xi\in\ext(r)$ if and only if $G_\xi$ is positive definite.
\end{proof}

\newcommand{\specialP}{\P^{\xi}}

\subsection{From extremals to POVMs, MICs and SICs}

We now single out those extremals in $\ext(r)$ whose associated rank-one operator basis can be normalised to a POVM on the dilation space. Let $\xi\in\ext(r)$ and define
\[
\mathbf e^\xi:=G_\xi^{-1}\mathbf 1\in\mathbb{R}^{d},
\]
where $\mathbf 1\in\mathbb{R}^{d}$ denotes the column vector of ones. Since $\{\Pi_n^\xi\}_{n=1}^{r^2}$ is a basis of $\mathcal B(\mathcal K)$, we can uniquely expand the identity in this basis, and it is easy to see that the resulting coefficients are the elements of $\mathbf e^\xi$:
\begin{equation}\label{compl}
\id_r=\sum_{n=1}^{r^2} e^\xi_n\,\Pi_n^\xi.
\end{equation}
Indeed, taking Hilbert--Schmidt inner products with the basis elements gives $G_\xi\,\mathbf e^\xi=\mathbf 1$. Thus $\mathbf e^\xi$ is the coordinate vector of the identity in the operator basis determined by $\xi$. In particular, $\mathbf e^\xi$ is independent of the choice of structure vectors and of the representative of the equivalence class of $\xi$. Taking the trace in \eqref{compl} gives
\begin{equation}\label{norm}
\sum_{n=1}^{r^2} e^\xi_n
=
\mathbf 1^TG_\xi^{-1}\mathbf 1
=
r.
\end{equation}

The signs of these unique coefficients determine whether the operator basis can be normalised to a measurement. In general, some of the coefficients $e^\xi_n$ may be zero or negative. If all $e_n^\xi\geq0$ (denoted $\mathbf e^\xi\geq0$), the nonzero operators $e_n^\xi\Pi_n^\xi$ form a rank-one POVM. If, moreover, every coefficient is strictly positive (denoted $\mathbf e^\xi>0$), all $r^2$ basis elements remain as nonzero effects and the resulting POVM is informationally complete. Motivated by these two cases, we define
\begin{equation}\label{eq:e-povm}
\Eobs
:=
\{\xi\in\ext(r)\mid\mathbf e^\xi\geq0\},
\end{equation}
and
\begin{equation}\label{eq:e-mic}
\Ep
:=
\{\xi\in\ext(r)\mid\mathbf e^\xi>0\}.
\end{equation}
Thus $\Eobs$ consists of the POVM-normalisable extremals, whereas $\Ep$ consists of the MIC-normalisable extremals.

For any $\xi\in\Eobs$, Eq.~\eqref{compl} defines a rank-one POVM on $\mathcal K$ by
\[
\M(n):=e^\xi_n\,\Pi_n^\xi,
\]
with outcome set
$\Omega_\M:=\{n\in\{1,\ldots,r^2\}\mid e^\xi_n\neq0\}$. Note that the structure vectors are unique up to unitary equivalence, so only the corresponding equivalence class of $\M$ is well-defined by $\xi$. Furthermore, if we change $\xi$ within its equivalence class, the structure vectors can only pick up an $n$-dependent phase factor, which does not affect $\M$. Accordingly, we let $\Mone$ denote the set of (unitary) equivalence classes of rank-one POVMs in dimension $r$; the following map is then well-defined:
\[
\Phi:\Eobs\to\Mone,
\qquad
\Phi(\xi)=\M.
\]

The inclusion $\Ep\subsetneq\Eobs$ is strict: examples in $\Eobs\setminus\Ep$ for $r=2$ are given in Eq.~\eqref{eq:non-mic-xi} of Sec.~\ref{sec:nonmic-family}. The subset $\Ep$ consists of (equivalence classes of) those extremals whose associated POVM is informationally complete. Indeed, extremality means that the projectors $\{\Pi_n^\xi\}_{n=1}^{r^2}$ span $\mathcal B(\mathcal K)$, while $\mathbf e^\xi>0$ ensures that none of the associated effects of $\M$ vanish. Thus the associated POVM is a rank-one MIC. Conversely, every rank-one MIC arises in this way. Denoting by $\MIC\subset\Mone$ the set of (equivalence classes of) rank-one MICs in dimension $r$, we therefore have the following correspondence (with a proof in Appendix~\ref{secA:prop-equiv}):

\begin{proposition}\label{prop:equiv}
The restriction of $\Phi$ to $\Ep$ is a bijection onto $\MIC$. Furthermore, for every $\xi\in\Ep$, with $\M=\Phi(\xi)$, one has
\begin{equation}\label{obsGIO}
\tr{\M(n)\M(m)}
=
e^\xi_n e^\xi_m\,(G_\xi)_{nm}\,,
\end{equation}
for all $n,m\in\{1,\ldots,r^2\}$.
\end{proposition}
Examples in $\Ep$ and their corresponding MICs are discussed in Sections \ref{sec:hw-extremals} and \ref{sec:semi-sic-extremals} for $r=2$.

A natural subclass of symmetric extremals within $\Ep$ are those with $e^\xi_n=1/r$ for all $n$. The corresponding rank-one MICs are known as \emph{unbiased} MICs \cite{debrota21}. Within this subclass, SICs appear as the most symmetric case.

\begin{proposition}\label{prop:sic-extemal-equiv}
Let $\xi\in\ext(r)$, with $\xi_{mn}=\ip{a_m}{a_n}$ and
$\Pi_n^\xi=|a_n\rangle\langle a_n|$ for $n=1,\ldots,r^2$. Then the following are
equivalent:
\begin{enumerate}
\item[\itlabel{i}] $(G_\xi)_{nm}=\frac{1}{r+1}$ for all $n\neq m$;
\item[\itlabel{ii}] $\{r^{-1}\Pi_n^\xi\}_{n=1}^{r^2}$ is a SIC-POVM.
\end{enumerate}
\end{proposition}

\begin{proof}
Since $(G_\xi)_{nm}=\tr{\Pi_n^\xi\Pi_m^\xi}$, (ii) immediately implies
(i).

Conversely, suppose that $(G_\xi)_{nm}=1/(r+1)$ for all $n\neq m$. Since $(G_\xi)_{nn}=1$ and $d=r^2$, every row of $G_\xi$ sums to $r$. Hence $G_\xi\mathbf 1=r\mathbf 1$. Since $\xi\in\ext(r)$, $G_\xi$ is invertible, and therefore $\mathbf e^\xi=(G_\xi)^{-1}\mathbf 1=\mathbf 1/r$. Hence $\sum_n r^{-1}\Pi_n^\xi=\id$. Together with $\tr{\Pi_n^\xi\Pi_m^\xi}=1/(r+1)$ for all $n\neq m$, this shows that $\{r^{-1}\Pi_n^\xi\}_{n=1}^{r^2}$ is a SIC-POVM.
\end{proof}

Thus an extremal coherence matrix $\xi\in\ext(r)$ satisfying condition (i) exists if and only if a SIC exists in dimension $r$. We therefore call such an extremal decoherence channel a SIC-extremal.

Finally, the POVM normalisation \eqref{compl} not only associates extremals with POVMs but also produces a canonical element of the compatibility region $\mathscr C_\xi$ defined in Eq.~\eqref{eq:c_xi}. For any $\xi\in\Eobs$, taking the auxiliary POVM $\F$ in Theorem~\ref{thm:gii} to be $\M=\Phi(\xi)$ gives $\xi_{nm}(j)=e^\xi_j\,\xi_{nj}\xi_{jm}$. Hence, we get a \emph{special incoherent observable} $\specialP\in\mathscr C_\xi$ with conditional probabilities
\begin{equation}\label{eq:special_incoherent}
p(j|n)=e^\xi_j\,(G_\xi)_{nj}\,,
\end{equation}
and outcome set $\Omega_{\P^\xi}:=\{j\in\{1,\dots,d\}\mid e^\xi_j\neq0\}$. Thus every $\xi\in\Eobs$ determines a (non-trivial) canonical incoherent observable in $\mathscr C_\xi$ arising directly from the extremal channel itself. We will see in Sec.~\ref{sec:jm-sic} that this observable plays a useful role in distinguishing SIC-extremals from other extremals in terms of joint measurability.

\section{Joint measurability under extremal decoherence}\label{sec:jm-ext}

\subsection{General characterisation of joint measurability}

We will derive a simple membership test for the compatibility region $\mathscr C_\xi$, defined in Eq.~\eqref{eq:c_xi}, for extremal coherence matrices $\xi\in\ext(r)$. The general criteria developed in \cite{kiukas22}, while applicable to any coherence matrix, fails to yield strong results in the extremal case. Here we use the associated rank-one operator frame of the extremal to obtain an exact characterisation of $\mathscr C_\xi$.

By the general theory (Theorem~\ref{thm:gii}), an incoherent observable $\P\in\mathscr I_d$ belongs to $\mathscr C_\xi$ if and only if there exists a POVM $\F$ on $\mathcal K$ such that
\begin{equation}\label{eq:ext-dilation-probabilities}
p(j|n)
=
\langle a_n|\F(j)|a_n\rangle
=
\tr{\Pi_n^\xi\F(j)}\,,
\end{equation}
for all $j,n$.

Now let $\xi\in\ext(r)$, so that $d=r^2$ and $\{\Pi_n^\xi\}_{n=1}^{d}$ is an operator basis of $\mathcal B(\mathcal K)$. Consequently, every Hermitian operator $\F(j)$ admits a unique expansion
\begin{equation}\label{eq:F}
\F(j)=\sum_{n=1}^{d}q_n(j)\Pi_n^\xi\,,
\end{equation}
with real coefficients $q_n(j)$. Substituting this expansion into \eqref{eq:ext-dilation-probabilities} gives
\[
p(j|n)
=
\sum_{m=1}^{d}(G_\xi)_{nm}q_m(j).
\]
Writing
\[
\mathbf p(j):=(p(j|1),\ldots,p(j|d))^T,
\,\,\,
\mathbf q(j):=(q_1(j),\ldots,q_d(j))^T,
\]
we therefore obtain
\begin{equation}\label{eq:pq-relation}
\mathbf p(j)=G_\xi\mathbf q(j).
\end{equation}
Since $G_\xi$ is invertible, the observable $\P$ uniquely determines the candidate effects through $\mathbf q(j)=G_\xi^{-1}\mathbf p(j)$. This gives the following criterion.

\begin{theorem}\label{thm:extr}
Let $\xi\in\ext(r)$ and let
$\P=\sum_n p(\cdot|n)|n\rangle\langle n|$ be an incoherent observable, with coefficient vectors $\mathbf p(j)$ as above. Choose any representative $\{\Pi_n^\xi\}_{n=1}^{r^2}$ of the associated operator basis and define
\begin{equation}
\F_\xi^\P(j)
:=
\sum_{n=1}^{r^2}
\bigl[G_\xi^{-1}\mathbf p(j)\bigr]_n\Pi_n^\xi.
\end{equation}
Then the following are equivalent:
\begin{itemize}
\item[\itlabel{i}] $\P\in\mathscr C_\xi$;
\item[\itlabel{ii}] $\F_\xi^\P(j)\geq0$ for every outcome $j$.
\end{itemize}
\end{theorem}

\begin{proof}
Suppose first that $\P\in\mathscr C_\xi$. By Theorem~\ref{thm:gii}, there exists a POVM $\F$ satisfying \eqref{eq:ext-dilation-probabilities}. Expanding each $\F(j)$ in the operator basis $\{\Pi_n^\xi\}$ and using \eqref{eq:pq-relation} gives $\mathbf q(j)=G_\xi^{-1}\mathbf p(j)$. Hence $\F(j)=\F_\xi^\P(j)$ for every $j$. Since $\F$ is a POVM, it follows that $\F_\xi^\P(j)\geq0$ for every $j$.

Conversely, suppose that $\F_\xi^\P(j)\geq0$ for every $j$. Since $\P$ is an observable, $\sum_j\mathbf p(j)=\mathbf 1$, and therefore, it follows from Eq.~\eqref{compl} that $\sum_j\F_\xi^\P(j)=\id_r$. Thus $\F_\xi^\P$ is a POVM. Moreover, its definition and \eqref{eq:pq-relation} give $\tr{\Pi_n^\xi\F_\xi^\P(j)}=p(j|n)$ for all $j,n$. The dilation characterisation in Theorem~\ref{thm:gii} therefore implies that $\P\in\mathscr C_\xi$.
\end{proof}

Thus, whereas for a non-extremal coherence matrix one must generally search over dilation-space POVMs satisfying the prescribed conditional probabilities, for $\xi\in\ext(r)$ the observable $\P$ determines a unique candidate POVM. Membership in $\mathscr C_\xi$ therefore reduces to checking positivity of its effects.

\subsection{Quantifying loss of incompatibility by volume}
\label{sec:volume-compatibility}

We now use the equivalence in Theorem~\ref{thm:extr} to evaluate the compatibility volume introduced in Eq.~\eqref{eq:rel-vol_IB} for extremals $\xi\in\mathcal E(r)$. We restrict the probe observables to those with a fixed number of outcomes. Hence, for $m\ge 2$, let $\mathscr I^{(m)}_{\!d}$ denote the set of $m$-outcome incoherent observables in dimension $d=r^2$. The quantity we consider is the compatibility volume of $\xi$ relative to the IB channel defined in \eqref{eq:rel-vol_IB}, namely $\vol(\xi;\mathscr I^{(m)}_{\!d})=\vol(\mathscr C_\xi[\mathscr I^{(m)}_{\!d}])/\vol(\mathscr I^{(m)}_{\!d})$.

We parametrise $\P\in\mathscr I^{(m)}_{\!d}$ by its probabilities $p(j|n)$, taking the entries with $j=1,\ldots,m-1$ as independent coordinates. For each $n$, the final probability is fixed by normalisation. Thus $\mathscr C_\xi[\mathscr I^{(m)}_{\!d}]$ can be treated as a subset of $\mathbb R^{(m-1)d}$. Similarly, let $\mathscr M^{(m)}_r$ be the set of $m$-outcome observables on $\mathcal K$, so that
\[
\mathscr M^{(m)}_r
\!:=\!
\left\{
(\F(1),\ldots,\F(m))
\;\middle|\;
\F(j)\geq 0,\ 
\sum_{j=1}^m \F(j)=\id
\right\}.
\]
We parametrise this set by the first $m-1$ effects, with the final effect fixed by normalisation. Thus $\mathscr M^{(m)}_r$ can also be treated as a subset of $\mathbb R^{(m-1)r^2}$. Since $d=r^2$, both sets are measured in real dimension $(m-1)d$.

To relate them, we fix any representative $\{\Pi_n^\xi\}_{n=1}^d$ of the operator basis associated with $\xi$. Any other representative is obtained by a common unitary conjugation, which acts orthogonally in the Hilbert--Schmidt geometry and therefore does not affect the volume calculation. We are now ready to determine the compatibility volume:

\begin{proposition}
\label{prop:volume-original-Cxi}
Let $\xi\in\ext(r)$ and let $m\ge2$. Then
\[
\vol(\xi;\mathscr I^{(m)}_{\!d})
=
c_{m,r}\cdot \bigl(\det G_\xi \bigr)^{(m-1)/2},
\]
where $c_{m,r}:=\vol(\mathscr M^{(m)}_r)/\vol(\mathscr I_{\!d}^{(m)})$ is independent of $\xi$.
\end{proposition}

\begin{proof}
Fix any representative $\{\Pi_n^\xi\}_{n=1}^d$ of the operator basis associated with $\xi$. Relative to this choice, Theorem~\ref{thm:extr} implies that an observable $\P\in\mathscr C_\xi[\mathscr I^{(m)}_{\!d}]$
corresponds to a unique POVM $\F_\xi^\P\in\mathscr M^{(m)}_r$, with $p(j|n)=\tr{\F_\xi^\P(j)\,\Pi_n^\xi}$. We use the first $m-1$ outcomes as independent variables.

Choose an orthonormal basis $H_1,\ldots,H_d$ of Hermitian operators, so that $\tr{H_\alpha H_\beta}=\delta_{\alpha\beta}$. Let $\Pi_n^\xi=\sum_{\alpha=1}^d b_{n\alpha}H_\alpha$, where $b_{n\alpha}\in \mathbb R$. If $\F_\xi^\P(j)=\sum_{\alpha=1}^d x_\alpha(j)H_\alpha$, and $\mathbf x(j)=(x_1(j),\ldots,x_d(j))^T$, then
\[
p(j|n)=\tr{\F_\xi^\P(j)\,\Pi_n^\xi}
=
\sum_{\alpha=1}^d b_{n\alpha}x_\alpha(j),
\]
and therefore $\mathbf p(j)=B\,\mathbf x(j)$, with $B=(b_{n\alpha})_{n,\alpha=1}^d$. Note that $B$ is invertible since $\Pi_1^\xi,\ldots, \Pi_d^\xi$ is also a basis.

Let $S_m\subset(\mathbb R^d)^{m-1}$ be the coordinate representation of $\mathscr M^{(m)}_r$, i.e. the set of tuples $(x(1),\ldots,x(m-1))$ corresponding to effects $\F_\xi^\P(1),\ldots,\F_\xi^\P(m-1)$, with
$\F_\xi^\P(m)=\id-\sum_{j=1}^{m-1}\F_\xi^\P(j)$. Under the above correspondence, the coordinate representation of $\mathscr C_\xi[\mathscr I^{(m)}_{\!d}]$ is the image of $S_m$ under the block-diagonal linear map
\[
(x(1),\ldots,x(m-1))
\longmapsto
(Bx(1),\ldots,Bx(m-1)).
\]
The determinant of this block-diagonal map is $|\det B|^{m-1}$. Therefore
\[
\vol(\mathscr C_\xi[\mathscr I^{(m)}_{\!d}])
=
|\det B|^{m-1}\vol(S_m),
\]
where the volume $\vol(\mathscr M^{(m)}_r)=\vol(S_m)$ is clearly independent of the orthonormal operator basis used to define $S_m$. It remains to compute $|\det B|$. Since
\[
(BB^T)_{nm}
=
\sum_{\alpha=1}^d b_{n\alpha}b_{m\alpha}
=
\tr{\Pi_n^\xi\Pi_m^\xi}
=
(G_\xi)_{nm},
\]
we have $(\det B)^2=\det(BB^T)=\det G_\xi $, which (together with the definition in \eqref{eq:rel-vol_IB}) yields the claim.
\end{proof}

Thus, for fixed $r$ and $m$, the volume of the set of observables that lose incompatibility depends on $\xi$ only through $\det G_\xi$. Hence, within the maximal-rank extremal regime, the loss of incompatibility, as quantified by compatibility volume, is completely determined by the Gram matrix of the associated operator frame. In Sec.~\ref{sec:jm-sic} we show that this loss is maximised by SIC-extremals.

\subsection{Joint measurability under MIC-extremals}

We now further specialise to coherence matrices $\xi\in\Ep$ associated with rank-one MICs on $\mathcal K$. This is the natural generalisation of the SIC case from Sec. \ref{sec:sic-preview}.

By Theorem~\ref{thm:extr}, the probabilities defining an incoherent observable $\P$ determine a candidate POVM $\F_\xi^\P$ on the MIC Hilbert space. Compatibility of $\P$ with the decohered observable is equivalent to positivity of all these candidate effects. After normalisation, each positive $\F_\xi^\P(j)$ defines a quantum state $\rho_j$. Using the MIC state-space representation recalled in Sec.~\ref{sec:mic-state-spaces}, its MIC probability vector is
\begin{align}\label{eq:xi-normalised}
\mathbf s^\xi_\P(j)&:=
\frac{E_\xi\,\mathbf p(j)}{({\mathbf e}^\xi)^T\mathbf p(j)}
\in\Delta_d,&
s^\xi_\P(j|n)&=\frac{e^\xi_n\,p(j|n)}{\sum_m e^\xi_m\,p(j|m)},
\end{align}
where $\mathbf p(j)=(p(j|1),\ldots,p(j|d))^T$ and
$E_\xi=\operatorname{diag}(e^\xi_1,\ldots,e^\xi_d)$. Thus $\F_\xi^\P(j)\ge0$ is equivalent to $\mathbf s^\xi_\P(j)$ belonging to the MIC state space $\MICs(\M)$.

There is also an equivalent probabilistic reformulation. Evaluating the same operators on an arbitrary state with MIC probability vector $\mathbf p\in\MICs(\M)$ produces the vector $PD_\xi^{-1}\mathbf p$, where $D_\xi=E_\xi G_\xi$. Compatibility is therefore equivalent to requiring that $PD_\xi^{-1}\mathbf p$ be a genuine probability distribution for every $\mathbf p\in\MICs(\M)$.  We summarise these equivalent conditions in the following theorem.

\begin{theorem}\label{mainthm}
Let $\xi\in\Ep$ with associated MIC $\M$ on $\mathcal{K}$. Let $\P=\sum_n p(\cdot|n)|n\rangle\langle n|$ be an incoherent observable on $\mathcal H$, with outcome set $\Omega_\P$, and let $P=(p(j|n))$ be its column-stochastic matrix. Then the following are equivalent:
\begin{enumerate}
\item[\itlabel{i}] $\P\in \mathscr C_\xi$;
\item[\itlabel{ii}] $PD_\xi^{-1}\mathbf p\in \Delta_{\Omega_\P}$ for every $\mathbf p\in \MICs(\M)$;
\item[\itlabel{iii}] $\mathbf s^\xi_\P(j)\in\MICs(\M)$ for all $j\in\Omega_\P$.
\end{enumerate}
\end{theorem}

\begin{proof}
By Theorem~\ref{thm:extr}, the condition $\P\in\mathscr C_\xi$ is equivalent to $\F_\xi^\P(j)\ge0$ for all $j$, where, by extremality of $\xi$,
\begin{equation}\label{eq:F_MIC_explicit}
\F_\xi^\P(j)=\sum_{n,m=1}^d (e^\xi_n)^{-1}(G_\xi^{-1})_{nm}\,p(j|m)\,\M(n).
\end{equation}

We now evaluate the MIC probability vector of $\sigma_j:=\F_\xi^\P(j)/\tr{\F_\xi^\P(j)}$, namely,
\[
\bigl[p_\M(\sigma_j)\bigr]_n=\frac{\tr{\M(n)\F_\xi^\P(j)}}{\tr{\F_\xi^\P(j)}}=\frac{e^\xi_n\,p(j|n)}{\sum_m e^\xi_m\,p(j|m)}\,,
\]
where we have used $\tr{\M(n)\F_\xi^\P(j)}=e^\xi_n\,\tr{\Pi_n^\xi\F_\xi^\P(j)}=e^\xi_n\,p(j|n)$ and $\tr{\F_\xi^\P(j)}=(\mathbf e^\xi)^T\mathbf p(j)$. Thus, $p_\M(\sigma_j)=\mathbf s^\xi_\P(j)$, with $\mathbf s^\xi_\P(j)$ defined in Eq.~\eqref{eq:xi-normalised}. We conclude that
\[
\F_\xi^\P(j)\ge0
\quad\Longleftrightarrow\quad
\sigma_j\in\mathcal S(\mathcal K)
\quad\Longleftrightarrow\quad
\mathbf s^\xi_\P(j)\in\MICs(\M).
\]
This proves the equivalence of (i) and (iii).

For the probabilistic reformulation, let $\mathbf p\in\MICs(\M)$ be arbitrary, and let $\sigma$ be the corresponding state, so that $p_n=\bigl[p_\M(\sigma)\bigl]_n=\tr{\M(n)\sigma}$. Defining $p(j):=\tr{\F_\xi^\P(j)\sigma}$ and using Eq.~\eqref{eq:F_MIC_explicit}, we obtain
\[
p(j)=
\sum_{n,m=1}^d (e_n^\xi)^{-1}(G_\xi^{-1})_{nm}\,p(j|m)\,p_n.
\]
Denoting $D_\xi=E_\xi G_\xi$, the vector with components $p(j)$, $j\in\Omega_\P$ is therefore $PD_\xi^{-1}\mathbf p$. Since $\F_\xi^\P(j)\ge0$ if and only if $\tr{\F_\xi^\P(j)\sigma}\ge0$ for all states $\sigma$, we conclude that
\[
\P\in\mathscr C_\xi
\quad\Longleftrightarrow\quad
PD_\xi^{-1}\mathbf p\in\Delta_{\Omega_\P}
\ \text{for every }\mathbf p\in\MICs(\M),
\]
proving the equivalence of (i) and (ii).
\end{proof}

We stress that the conceptual discussion following Theorem \ref{motivationthm} in Section \ref{sec:sic-preview} (the SIC-case) applies also here; joint measurability is in one-to-one correspondence with the existence of a sequential post-MIC measurement model reproducing the conditional probabilities of the incoherent observable.

To describe the joint measurement realisation arising from this structure, we first note that a rank-one MIC $\M(n)=e_n^\xi\Pi_n^\xi$ naturally converts quantum states on the MIC Hilbert space $\mathcal K$ into probability distributions on a set of labels $n$. Indeed, measuring $\M$ on a state $\sigma$ produces the probabilities $p_n=\tr{\M(n)\sigma}$, which may be reinterpreted as the diagonal state $\rho=\sum_n p_n|n\rangle\langle n|$ on the Hilbert space $\mathcal H$. In this way, the MIC provides a link between quantum states on $\mathcal K$ and diagonal states on $\mathcal H$. This link also generates both marginals of our joint measurement problem. A subsequent measurement $\F^\P_\xi$ on $\mathcal K$ defines conditional probabilities $p(j|n)=\tr{\F^\P_\xi(j)\Pi_n}$, and hence an incoherent observable $\P(j)=\sum_n p(j|n)|n\rangle\langle n|$ on $\mathcal H$. At the same time, a coherent observable $\Q$ on $\mathcal H$ is distorted by the overlap matrix $\xi_{nm}=\langle a_n|a_m\rangle$ of the MIC post-measurement states, giving the decohered observable $\xi\circ\Q$. 

To make this explicit, we embed the space $\mathcal H$ into the joint system $\mathcal H\otimes \mathcal K$ via the isometry $V$ defined in Eq. \eqref{eq:isometry} which couples the MIC post-measurement state $|a_n\rangle$ with the corresponding incoherent basis vector. In this way we can measure both $\F^\P_\xi$ and $\Q$ simultaneously, and reinterpret the outcome distribution again in terms of the label states only. This leads to the POVM $\G$ defined in Eq. \eqref{eq:jm}. Ignoring outcome $i$ of the coherent measurement we obtain the observable $\P$ defined above, while the other marginal is a version of $\Q$ distorted by decoherence. In this way, the MIC ``acts'' on the label space in two different ways: \emph{incoherently} through the diagonal observable $\P$ arising from the subsequent measurement $\F^\P_\xi$, and \emph{coherently} via observable $\xi \circ \Q$ arising from the coherent measurement $\Q$ carried out under decoherence.

We conclude this section by recalling the special incoherent observable $\P^\xi$ defined in \eqref{eq:special_incoherent}, obtained by taking the auxiliary POVM $\F$ in Theorem~\ref{thm:gii} to be the MIC $\M$ associated with $\xi$. In this case,
\[
p(j|n)=e^\xi_j(G_\xi)_{nj}=e^\xi_j|\langle a_n|a_j\rangle|^2.
\]
Since $(E\mathbf p^\xi(j))_n=e^\xi_j\,\tr{\Pi_j^\xi\,\M(n)}$ and $(\mathbf e^\xi)^T\mathbf p^\xi(j)=e^\xi_j$, Eq.~\eqref{eq:xi-normalised} gives
\[
\mathbf s^\xi_\P(j)=p_\M(\Pi_j^\xi)\,,
\]
and hence $s^\xi_\P(j|n)=\tr{\Pi_j^\xi\,\M(n)}=e^\xi_n\tr{\Pi_j^\xi\,\Pi_n^\xi}$. Thus the rows of the column-stochastic matrix of $\P^\xi$ are the MIC probability vectors of the rank-one projectors defining the MIC. These are the \emph{basis distributions} studied in the QBist literature on MIC state spaces \cite{fuchs13,appleby11}.

Since each $\Pi_j^\xi$ is a pure state, the corresponding vector $\mathbf s^\xi_\P(j)=p_\M(\Pi_j^\xi)$ is an extreme point of the MIC state space $\mathcal P(\M)$, and hence lies on its boundary \cite{appleby11b,fuchs13}. It follows that $\P^\xi$ determines a canonical simplex
\begin{equation}\label{eq:Pxi-simplex}
\mathcal P^\xi(\M):=\operatorname{conv}\{\mathbf s^\xi_\P(1),\dots,\mathbf s^\xi_\P(d)\}\subseteq \mathcal P(\M),
\end{equation}
whose vertices are boundary points of the MIC state space. For the SIC case, as discussed in Sec.~\ref{sec:jm-sic}, the simplex is regular.

\section{Joint measurability under SIC-decoherence}\label{sec:jm-sic}

We now consider extremal channels $\xi\in\Ep$ whose associated MIC $\M=\Phi(\xi)$ is a SIC; this highly symmetric case turns out to be special also in terms of our joint measurability problem.

The coherence matrix and structure vectors satisfy
\begin{equation}\label{eq:sic-channel}
|\xi_{mn}|^2=|\langle a_m|a_n\rangle|^2=\tr{\Pi_m^\xi\Pi_n^\xi}=\frac{1}{r+1}\,,
\end{equation}
for all $m\neq n$. In this case, the MIC reconstruction formula becomes explicit, the compatibility criterion acquires the QBist form of the Born rule, and the resulting compatibility region has several distinctive features.

Recall that in the general MIC setting, $\P\in\mathscr{C}_\xi$ is equivalent to $PD_\xi^{-1}\mathbf p\in\Delta_{\Omega_\P}$ for every $\mathbf p\in\MICs(\M)$, where $P=(p(j|n))$ is the column-stochastic matrix defining $\P$ and $D_\xi=E_\xi G_\xi$. In the SIC case the dual frame is explicit (cf. Eq.~\eqref{eq:dual_sic}), and one obtains $D_\xi^{-1}=(r+1)\id-r^{-1} |\mathbf 1\rangle\langle \mathbf 1|$. Substituting this into the general criterion (condition (ii) in Theorem \ref{mainthm}) yields the probabilistic condition (v) of Theorem~\ref{motivationthm}.

This makes the connection with the classical law of total probability especially transparent. In the general MIC case, compatibility is governed by the modified probability vector $D_\xi^{-1}\mathbf p$, whereas in classical probability one would simply have $(P\mathbf p)_j=\sum_n p(j|n)p_n$. Thus the matrix $D_\xi^{-1}$ quantifies the departure from the classical law of total probability. In the SIC case, $D_\xi$ is brought closest to the identity \cite{debrota20,debrota21}, and this departure takes the particularly simple QBist form highlighted in Section \ref{sec:sic-preview}.

Furthermore, the criterion $\mathbf s^\xi_\P(j)\in \mathcal P(\M)$ of Theorem~\ref{mainthm} is equivalent to
\begin{equation}\label{eq:SIC-positivity}
\sum_{n=1}^{r^2}
\left((r+1)s^\xi_\P(j|n)-\frac1r\right)\Pi_n^\xi
\ge 0\,,
\end{equation}
for all $j\in\Omega_\P$, as seen in Eq.~\eqref{eq:SIC-reconstruction2}. Given an incoherent observable $\P$, the vector $\mathbf s^\xi_\P(j)$ defined in Eq.~\eqref{eq:xi-normalised} corresponds to the normalised $j$-th row of the column-stochastic matrix defining $\P$. These two observations lead to conditions (iv) and (vi) of Theorem~\ref{motivationthm}.

\subsection{SIC-extremals maximise loss of incompatibility}\label{subsec:sic_max}
Interestingly, we can single out SIC-extremal channels among \emph{all} extremals $\xi\in\ext(r)$ by comparing their compatibility volumes. We first use the IB-normalised volume $\vol(\xi;\mathscr I^{(m)}_{\!d})$ defined in Eq.~\eqref{eq:rel-vol_IB}, where $\mathscr I^{(m)}_{\!d}$ is the set of $m$-outcome incoherent observables. By Proposition~\ref{prop:volume-original-Cxi}, for fixed $r$ and $m$, the dependence of $\vol(\xi;\mathscr I^{(m)}_{\!d})$ on $\xi$ is entirely through $\det G_\xi$. A result in \cite{debrota20} provides a bound on this quantity, and shows that it is saturated only when the operator frame corresponds to a SIC.

\begin{lemma}
\label{lem:det-bound}
Let $\xi\in\ext(r)$. Then
\[
\det G_\xi
\le
r\left(\frac{r}{r+1}\right)^{r^2-1}\,,
\]
with equality if and only if $\xi$ is SIC-extremal.
\end{lemma}

\begin{proof}
This follows from \cite[Lemma~3 and Eqs.~(A30)--(A31)]{debrota20}, applied to
$(G_\xi)_{nm}=\tr{\Pi_n^\xi\Pi_m^\xi}$. We reproduce the short argument in
Appendix~\ref{app:det-bound} for completeness.
\end{proof}

Suppose now that there exists a SIC-extremal $\xi_{\rm SIC}\in\ext(r)$. For comparisons within $\mathcal{E}(r)$, it is natural to normalise by the SIC-extremal rather than by the IB channel, due to the special role of the SIC-case. Thus, in the notation of Eq.~\eqref{eq:rel-vol}, we consider the compatibility volume of $\xi$ relative to a SIC-extremal, given by
\begin{equation}\label{eq:sic-rel-vol}
\vol(\xi\,;\xi_{\rm SIC},\mathscr I^{(m)}_{\!d})
=
\frac{\vol(\mathscr C_\xi[\mathscr I^{(m)}_{\!d}])}
{\vol(\mathscr C_{\xi_{\rm SIC}}[\mathscr I^{(m)}_{\!d}])}.
\end{equation}

Using the above lemma and Proposition \ref{prop:volume-original-Cxi}, this volume simplifies as follows:
\begin{theorem}
\label{thm:volume-comparison-sic}
Let $\xi\in\ext(r)$ and suppose there exists a SIC-extremal $\xi_{\rm SIC}\in\ext(r)$. Then
\begin{equation}\label{eq:relative}
\vol(\xi\,;\xi_{\rm SIC},\mathscr I^{(m)}_{\!d})=
\left[
\frac{1}{r}
\left(\frac{r+1}{r}\right)^{\!r^2-1}
\!\det G_\xi
\right]^{\frac{m-1}{2}}.
\end{equation}
Furthermore, $0< \vol(\xi\,;\xi_{\rm SIC},\mathscr I^{(m)}_{\!d})\leq 1$ with equality in the upper bound if and only if $\xi$ is SIC-extremal.
\end{theorem}
\begin{proof}
Using Proposition~\ref{prop:volume-original-Cxi} we get
\begin{equation}\label{eq:relative-det}
\vol(\xi\,;\xi_{\rm SIC},\mathscr I^{(m)}_{\!d})
=
\left(
\frac{\det G_\xi}
{\det G_{\xi_{\rm SIC}}}
\right)^{\frac{m-1}{2}}\,.
\end{equation}
Now Lemma~\ref{lem:det-bound} gives $\det G_{\xi_{\rm SIC}}=r\left(r/(r+1)\right)^{r^2-1}$ so after substitution we get Eq. \eqref{eq:relative}. Furthermore, we conclude from Lemma~\ref{lem:det-bound} that $0<\vol(\xi\,;\xi_{\rm SIC},\mathscr I^{(m)}_{\!d})\le 1$, with equality if and only if $\xi$ is SIC-extremal.
\end{proof}
Thus, SIC-extremals maximise the compatibility volume among all maximal-rank extremal decoherence channels:
\[
\vol(\mathscr C_\xi[\mathscr I^{(m)}_{\!d}])\leq \vol(\mathscr C_{\xi_{\rm SIC}}[\mathscr I^{(m)}_{\!d}])\,.
\]
We can infer from Eq.~\eqref{eq:relative} that the number of outcomes $m$ amplifies the distinction between extremals. In particular, if $\xi$ is not SIC-extremal, then $\vol(\xi\,;\xi_{\rm SIC},\mathscr I^{(m)}_{\!d})$ decays exponentially in $m$.

We stress that the SIC-case obviously does not maximise the compatibility volume among \emph{all} decoherence channels, as the region typically expands with added classical noise. By restricting to extremals we eliminate this possibility, thereby uncovering the special role of the SIC-case as the ``most classical'' extremal in the sense of joint measurability.

At the opposite, ``furthest from classical'' end lie extremals that preserve a large amount of incompatibility. For fixed $r$ and $m$, the relative compatibility volume approaches zero as $\det G_\xi\to0$. Geometrically, $\det G_\xi$ is the squared Hilbert--Schmidt volume spanned by $\{\Pi_n^\xi\}_{n}$, therefore the limit corresponds to the operator frame becoming nearly linearly dependent. The channel consequently approaches the boundary of the maximal-rank extremal regime, where the frame ceases to be an operator basis. This behaviour occurs in all three families studied in Sec.~\ref{sec:qubit-mics}.

It is useful to compare the compatibility volume with the robustness quantifier of incompatibility preservability introduced in \cite{hsieh25}. As explained in Sec.~\ref{subsec:generalcompreg}, this measures the minimum amount of mixing with an arbitrary auxiliary channel required to make a channel incompatibility breaking. Since its exact evaluation is difficult, we introduce a simpler version in Appendix~\ref{app:resource-preservability}, in which the auxiliary channel is restricted to the class of decoherence channels. Then, for extremals in $\ext(r)$, the robustness is minimised by channels associated with unbiased rank-one MICs. It therefore places SIC-extremals among the least incompatibility-preserving extremals, but does not single them out.

Indeed, the robustness depends only on the maximum eigenvalue of $\xi$ and therefore cannot distinguish operator frames sharing the same maximum eigenvalue. By contrast, the compatibility volume depends on $\det G_\xi$ and is therefore sensitive to the collective geometry. In particular, it uniquely identifies SIC-extremals as its maximisers. The Heisenberg--Weyl extremals considered in the next section, whose associated MICs are unbiased, makes this distinction especially clear: all its members have the same robustness, whereas their relative compatibility volumes range from one to values arbitrarily close to zero. Thus, even at the scalar level, the compatibility volume is strictly more discriminating on this family, while the full compatibility region retains still finer information about how incompatibility is lost.

\subsection{The depolarised basis observable and comparison to non-extremal decoherence}
\label{subsec:sic-regular}

Here we proceed to give an explicit example of the distinctive behaviour of SIC-extremals, which also allows us to make a meaningful comparison to non-extremal decoherence with the same rates. To do this we restrict the probe class to a simple one parameter family of incoherent observables: let
\[
\mathscr I^{\rm dep}_{\!d}
:=
\left\{
\P_\alpha
\;\middle|\;
\alpha\in[0,1]
\right\}
\subseteq\mathscr I_{\!d}\,,
\]
where
\begin{equation}\label{eq:line-observable}
\P_{\alpha}:=\alpha \P_0 +(1-\alpha) r^{-2}\id\,,
\end{equation}
with $\P_0(j)=\kb{j}{j}$ for $j=1,\ldots,d=r^2$. Thus $\P_\alpha$ is the depolarised incoherent basis observable, and $\mathscr I^{\rm dep}_{\!d}$ is the line segment joining the trivial uniform observable to the incoherent basis observable.

We write the restricted compatibility region as
\begin{equation}\label{eq:alpha-thresh}
\mathscr C_\xi[\mathscr I^{\rm dep}_{\!d}]
=
\left\{
\P_\alpha
\;\middle|\;
0\leq\alpha\leq\alpha^*(\xi)
\right\}\,,
\end{equation}
where $\alpha^*(\xi)$ is the compatibility threshold, i.e., the largest value of $\alpha$ for which $\P_\alpha$ belongs to $\mathscr C_\xi$. Below we give a general upper bound for $\alpha^*(\xi)$ when $\xi\in\ext(r)$, and show that it is saturated if and only if the extremal corresponds to a SIC.

\begin{proposition}\label{prop:sic-maximal-line}
Let $\xi\in\ext(r)$. Then
\begin{equation}\label{eq:rank-s-ineq}
\alpha^*(\xi)\le\frac1{r+1}\,,
\end{equation}
with equality if and only if $\xi$ is SIC-extremal.
\end{proposition}
We give the proof in Appendix~\ref{secA:SIC-line-equiv}. It follows that, among all maximal-rank extremals, SIC-extremals are those for which the compatibility region $\mathscr C_\xi[\mathscr I^{\rm dep}_{\!d}]$ extends furthest along the depolarised-basis line.

We note that at the boundary $\alpha=1/(r+1)$, the observable $\P_{1/(r+1)}$ coincides with the special incoherent observable $\P^\xi$ introduced in Eq.~\eqref{eq:special_incoherent} when $\xi$ is a SIC-extremal, i.e. $\P^\xi=\P_{1/(r+1)}$. Recall the canonical simplex $\mathcal P^\xi(\M)$ defined in Eq.~\eqref{eq:Pxi-simplex} via the points $\mathbf s^\xi(j)$ of $\P^\xi$. For a SIC, these vertices form a regular simplex. The same regular simplex can also be defined independently of $\xi$ by taking as vertices the rows of the stochastic matrix defining $\P_{1/(r+1)}$. Proposition~\ref{prop:sic-maximal-line} implies that this regular simplex is contained in the MIC state space $\mathcal P(\M)$ if and only if $\M$ is a SIC.

At the end of the preceding subsection we remarked that non-extremals can trivially have larger compatibility volumes than the SIC-extremal. However, there is one interesting case of comparison, namely to the non-extremal uniform coherence matrix $\xi_{nm}=1/\sqrt{1+r}$ for $n\neq m$, which has the same constant off-diagonal damping amplitudes in dimension $d=r^2$ but no phase factors. For this non-extremal channel, the compatibility threshold of the same restricted probe class $\mathscr I^{\rm dep}_{\!d}$ is known to be
$\alpha=g_{r^2}(1/\sqrt{1+r})$ \cite{kiukas22}, where
\[
g_d(x)=\frac 1d \left((d-2)(1-x)+2\sqrt{1-x}\sqrt{1+(d-1)x}\right).
\]
This exceeds $(1+r)^{-1}$ for all $r\geq2$. Hence, along the depolarised-basis line, the SIC-extremal destroys less incompatibility than the corresponding uniform non-extremal channel: a larger portion of the line remains outside the compatibility region. Interestingly, we have $g_{r^2}(1/\sqrt{1+r})\to1$, whereas $(r+1)^{-1}\to0$ as $r\to\infty$, showing that asymptotically on the system dimension, the effect becomes dichotomic: the SIC-extremal preserves all incompatibility while the uniform decoherence destroys it. Hence, even though the SIC-extremal is ``closest to classical'' among all other (maximal-rank) extremals, it is still ``infinitely'' better at preserving incompatibility compared to the phase-insensitive uniform case with the same decoherence rates. The critical role of the phase factors highlights the intricate nature of quantum decoherence, revealed here through joint measurability.

\subsection{SIC existence as a joint measurability problem}

A SIC is conjectured to exists in every finite dimension $r$ \cite{zauner99}. We now reformulate the existence of SICs in terms of the joint measurability of ``noisy'' versions of the sharp mutually unbiased pair $(\P_0,\Q_{\rm MUB})$. Recall $\P_0$ is the incoherent basis observable in the notation of Eq.~\eqref{eq:line-observable} and $\P_{1/(r+1)}$ is its depolarised form at $\alpha=1/(r+1)$.

\begin{theorem}
\label{thm:sic-existence-mub}
A SIC exists in dimension $r$ if and only if there exists an extremal coherence matrix $\xi\in\ext(r)$ such that $\P_{1/(r+1)}$ and
$\xi\circ\Q_{\rm MUB}$ are jointly measurable observables in dimension $r^2$.
\end{theorem}

\begin{proof}
Suppose first that a SIC exists in $\mathbb C^r$. Let $\xi$ be the corresponding SIC-extremal on $\mathbb C^{r^2}$. By Proposition~\ref{prop:sic-maximal-line}, $\alpha^*(\xi)=1/(r+1)$, and hence
$\P_{1/(r+1)}\in\mathscr C_\xi$. Therefore $\P_{1/(r+1)}$ and $\xi\circ\Q_{\rm MUB}$ are jointly measurable.

Conversely, suppose that there exists $\xi\in\ext(r)$ such that $\P_{1/(r+1)}$ and $\xi\circ\Q_{\rm MUB}$ are jointly measurable. Thus, $\P_{1/(r+1)}\in\mathscr C_\xi$, and therefore $\alpha^*(\xi)\ge 1/(r+1)$. On the other hand, Proposition~\ref{prop:sic-maximal-line} gives
$\alpha^*(\xi)\le 1/(r+1)$. Hence $\alpha^*(\xi)=1/(r+1)$. By the equality statement in Proposition~\ref{prop:sic-maximal-line}, $\xi$ is SIC-extremal. Therefore its structure vectors form a SIC in $\mathbb C^r$.
\end{proof}

In this way, the existence of a SIC is encoded in the possibility of destroying incompatibility between a pair of ``noisy'' mutually unbiased bases via depolarisation and extremal decoherence. This connects two of the most celebrated measurements in quantum theory, both of which have their own unresolved existence question \cite{fuchs17,mcnulty26}.

The volume result in Theorem~\ref{thm:volume-comparison-sic} gives an alternative formulation of SIC existence. A SIC exists in $\mathbb C^r$ if and only if there exists $\xi\in\ext(r)$ such that, for some, and hence every, $m\ge2$, the IB-normalised compatibility volume satisfies
\[
\vol(\xi;\mathscr I^{(m)}_{\!d})
=
c_{m,r}\cdot
\left[
r\left(\frac{r}{r+1}\right)^{r^2-1}
\right]^{\frac{m-1}{2}},
\]
where $c_{m,r}$ is the constant appearing in Proposition~\ref{prop:volume-original-Cxi}.

\section{Special extremals for $r=2$}
\label{sec:qubit-mics}

We now consider the smallest Hilbert space for which non-unitary extremal decoherence channels exist, namely $\mathcal H\simeq\mathbb C^4$. This case already has a relatively rich structure, provided by the qubit dilation space; we have $d=4=r^2$ with $\mathcal K\simeq\mathbb C^2$. We take $\xi\in\ext(2)$ and analytically describe the compatibility region $\mathscr C_\xi$ for several families of extremals, including MIC and non-MIC cases. 

\subsection{General considerations}
The positivity condition appearing in Theorem~\ref{thm:extr} can now be written in a particularly simple form, since for any $F\in M_2(\mathbb C)$ we have 
\[
F\geq 0
\,\,\,\Longleftrightarrow\,\,\,
\tr{F}\geq 0
\ \text{ and }\
\tr{F}^2\geq \tr{F^2}.
\]
Note that here the second condition is equivalent to $\det F\geq 0$. Denoting $\mathbf q(j):=G_\xi^{-1}\mathbf p(j)$, we have
\[
\tr{\F_\xi^\P(j)}=\mathbf 1^T\mathbf q(j),
\qquad
\tr{\F_\xi^\P(j)^2}=\mathbf q(j)^T G_\xi\,\mathbf q(j),
\]
so it follows that
\begin{equation}\label{eq:qubit-quad-q}
\F_\xi^\P(j)\ge0
\quad\Longleftrightarrow\quad
(\mathbf 1^T\mathbf q(j))^2-\mathbf q(j)^T G_\xi\,\mathbf q(j)\ge0.
\end{equation}
Rewriting in terms of the coefficient vectors $\mathbf p(j)=G_\xi\mathbf q(j)$, we obtain the following quadratic criterion.
\begin{proposition}
\label{prop:qubit-K-from-G}
Let $\xi\in\ext(2)$ and define
\begin{equation}\label{eq:K-matrix}
K_\xi:=G_\xi^{-1}(\mathbf 1\mathbf 1^T-G_\xi)G_\xi^{-1}.
\end{equation}
Let $\P\in \mathscr{I}_4$ with coefficient vectors $\mathbf p(j)$. Then $\P\in\mathscr C_\xi$ if and only if
\[
\mathbf p(j)^T K_\xi\,\mathbf p(j)\ge0
\qquad
\text{for all outcomes }j.
\]
\end{proposition}

Thus for $d=4$, the compatibility region is governed entirely by the Gram matrix $G_\xi$. Next, we show that for MICs, this criterion defines an ellipsoid for each measurement outcome.

Consider therefore extremals $\xi\in\mathcal E_{>}(2)$ with an associated MIC
$\M=\Phi(\xi)$, as defined in Sec.~\ref{sec:ext-dec}. We need to study the matrix $G_\xi$; noting first that $G_\xi\,\mathbf e^\xi=\mathbf 1$, it is convenient to consider the rescaled matrix
\[
H_\xi:=\sqrt E_\xi\,G_\xi\sqrt E_\xi,
\]
where $E_\xi=\operatorname{diag}(e^\xi_1,\dots,e^\xi_d)$. Note that now each $e^\xi_n$ is nonzero (as $\xi\in\mathcal E_{>}(2)$), so $H_\xi$ is real symmetric and positive definite. Therefore, it admits a real orthonormal eigenbasis. Moreover, $\sqrt{\mathbf e^\xi}$ is always an eigenvector of $H_\xi$ with eigenvalue 1, where $\sqrt{\mathbf e^\xi}$ denotes the vector with components $\sqrt{e^\xi_n}$. Since $\|\sqrt{\mathbf e^\xi}\|^2=\sum_n e^\xi_n=2$, we therefore have the spectral decomposition
\[
H_\xi=\frac12\sqrt{\mathbf e^\xi} [\sqrt{\mathbf e^\xi}]^T+\sum_{n=1}^{3}\lambda_n \mathbf v_n \mathbf v_n^T,
\]
where $\lambda_n\in(0,1)$ are the remaining three eigenvalues, and $\mathbf v_n$ are the corresponding (real) orthonormal eigenvectors. (Here we suppress the $\xi$-dependence.) In particular, each $\mathbf v_n$ is orthogonal to $\sqrt{\mathbf e^\xi}$.

We now define (no longer necessarily orthonormal) vectors  $\mathbf u_n$ by
\begin{equation}\label{eq:v-to-u}
\mathbf u_n:=\sqrt {E_\xi}\,\mathbf v_n.
\end{equation}
Since $G_\xi^{-1}=\sqrt {E_\xi}\,H_\xi^{-1}\sqrt{E_\xi}$, we have
\begin{equation}\label{xiabsdecomp}
G_\xi^{-1}=\frac 12\mathbf e^\xi[\mathbf e^\xi]^T+\sum_{n=1}^{3}\frac{1}{\lambda_n}\mathbf u_n\mathbf u_n^T.
\end{equation}
Theorem~\ref{thm:extr}, together with the decomposition
\eqref{xiabsdecomp}, then gives a simple characterisation of the compatibility
region $\mathscr C_\xi$ in terms of ellipsoids.
\begin{proposition}\label{ellipses}
Let $\xi\in \mathcal E_{>}(2)$, and let $\lambda_n$ and $\mathbf u_n$ be as
above. Let $\P\in \mathscr I_4$ with coefficient vectors
$\mathbf p(j)$. Define the vectors $\mathbf x_\xi^\P(j)\in \mathbb R^{3}$ by the coordinates
\begin{equation}\label{eq:x-coordinates}
[\mathbf x_\xi^\P(j)]_n:=\sqrt{2}\,
\frac{\mathbf u_n^T\mathbf p(j)}
{[\mathbf e^\xi]^T\mathbf p(j)},
\qquad n=1,2,3.
\end{equation}
Then
\[
\mathscr C_\xi
=
\left\{
\P
\;\middle|\;
\mathbf x_{\xi}^\P(j)\in\mathscr E_\xi
\text{ for each }j
\right\}\,,
\]
where
\[
\mathscr E_\xi
:=
\left\{
\mathbf x\in \mathbb R^{3}
\ \middle|\
\sum_{n=1}^{3}\frac{x_n^2}{\lambda_n}\leq 1
\right\}.
\]
\end{proposition}
The proof is given in Appendix~\ref{app:r=2}. Note that the $\mathbf u_n$ also depend on $\xi$ (we have not indicated this explicitly) but are not uniquely determined by it; however, the compatibility region does not depend on the choice. Below, the extremal $\xi$ will be fixed in each case, so we will suppress the $\xi$-dependence in $\mathbf x^\P_\xi(j)$. Also note that any trivial incoherent observable $\P(j)=p_j\id$ has $\mathbf p(j)=p_j\mathbf 1$, so
\[
\mathbf u_n^T\mathbf p(j)
=
p_j\mathbf v_n^T\sqrt{\mathbf e^\xi}
=0.
\]
Therefore each trivial $\P$ corresponds to the centre of the ellipsoid $\mathscr E_\xi$, which in particular lies in its interior.

Before proceeding to consider interesting special classes of extremals, we discuss briefly the indexing of the relevant vectors. As we consider $\mathcal H\simeq\mathbb C^4\simeq \mathbb C^2\otimes\mathbb C^2$, it is natural to use the two-qubit notation
$$
|00\rangle, |01\rangle, |10\rangle, |11\rangle
$$
for the incoherent basis, so the label set is $\Omega:=\{0,1\}^2$, and incoherent observables are specified by vectors $\mathbf p(j) = (p_{00}(j), p_{01}(j), p_{10}(j), p_{11}(j))^T$. We then use a similar convention for the vector $\mathbf x^\P(j)$ in Prop. \ref{ellipses}, so that
$\mathbf x^\P(j) =(x_{01}(j),x_{10}(j),x_{11}(j))$, and the corresponding eigenvalues are $\lambda_{01}, \lambda_{10}, \lambda_{11}$. This indexing is natural for the Heisenberg-Weyl extremals in the next subsection, but we use it throughout for the sake of consistency and comparison.

Finally, we introduce a natural symmetric subclass of incoherent observables we use to probe the extremals: we define $\mathscr I^{\rm sc}_{\!4}\subseteq\mathscr I_{\!4}$ to be the set of \emph{shift-covariant} incoherent observables $\P_{\mathbf r}$ of the form
\begin{equation}\label{eq:inc-sh}
\P_{\mathbf r}(n,m)
=
\sum_{(k,l)\in\Omega}
r_{n\ominus k,\;m\ominus l}\,
|kl\rangle\langle kl|\,,
\end{equation}
for $(n,m)\in\Omega$, where $\mathbf r=(r_{00},r_{01},r_{10},r_{11})\in\Delta_4$ is a seed distribution and $\ominus$ denotes subtraction modulo $2$. The corresponding compatibility region is then $\mathscr C_\xi[\mathscr I^{\rm sc}_{\!4}]$.

It turns out to be convenient to use the Fourier coordinates $\hat r_{nm}:=\sum_{(k,l)\in\Omega}(-1)^{kn+lm}r_{kl}$, where the coefficient $\hat r_{00}$ is fixed by normalisation, namely $\hat r_{00}=1$. Therefore, the relevant coordinates form 3-vectors
\begin{equation}\label{fouriercoords}
\hat{\mathbf r}= (\hat r_{01},\hat r_{10}, \hat r_{11}),
\end{equation}
which we will use below to visualise the relevant compatibility regions.

\subsection{Heisenberg--Weyl extremals}
\label{sec:hw-extremals}

We first apply Proposition~\ref{ellipses} to the class of extremals associated
with the qubit Heisenberg--Weyl MICs introduced in \cite{ariano04}. These MICs are generated by applying the unitaries of the Heisenberg-Weyl representation to a single ``seed vector''; see Appendix~\ref{app:hw-extremals} for details. The resulting coherence matrix depends on two parameters $0<s<1$ and $0<\vartheta<\pi/2$; writing $t=(1+s^2)^{-1/2}$, we have
\begin{equation}\label{eq:hw-extremal-xi}
\xi
=
t^{2}
\begin{pmatrix}
1+s^{2} & 1-s^{2} & 2s\cos\vartheta & 2 i s\sin\vartheta \\
1-s^{2} & 1+s^{2} & 2 i s\sin\vartheta & 2s\cos\vartheta \\
2s\cos\vartheta & -2 i s\sin\vartheta & 1+s^{2} & -(1-s^{2}) \\
-2 i s\sin\vartheta & 2s\cos\vartheta & -(1-s^{2}) & 1+s^{2}
\end{pmatrix}.
\end{equation}
 It follows that $\mathbf e^\xi=\mathbf 1/2$, that is, the MIC is unbiased, and hence
$H_\xi=\frac12 G_\xi$. Clearly, $H_\xi$ is a block-circulant matrix with circulant blocks, and is therefore diagonal in the Fourier (Hadamard) basis
\begin{equation}\label{eq:hadamard-basis}
\begin{pmatrix}
\mathbf v_{00} & \mathbf v_{01} & \mathbf v_{10} & \mathbf v_{11}
\end{pmatrix}
=
\frac12
\begin{pmatrix}
1 &  1 &  1 &  1 \\
1 & -1 &  1 & -1 \\
1 &  1 & -1 & -1 \\
1 & -1 & -1 &  1
\end{pmatrix},
\end{equation}
with eigenvalues $\lambda_{00}=1$ and
\begin{align}
\lambda_{01}(s,\vartheta)&=4t^4s^2\cos^2\vartheta,\nonumber\\
\lambda_{10}(s,\vartheta)&=t^4(1-s^2)^2,\label{eq:HW-evalues}\\
\lambda_{11}(s,\vartheta)&=4t^4s^2\sin^2\vartheta.\nonumber
\end{align}
This gives us the decomposition \eqref{xiabsdecomp} of $G_\xi^{-1}$.
\begin{figure}[t]
  \centering
  \includegraphics[width=0.85\linewidth]{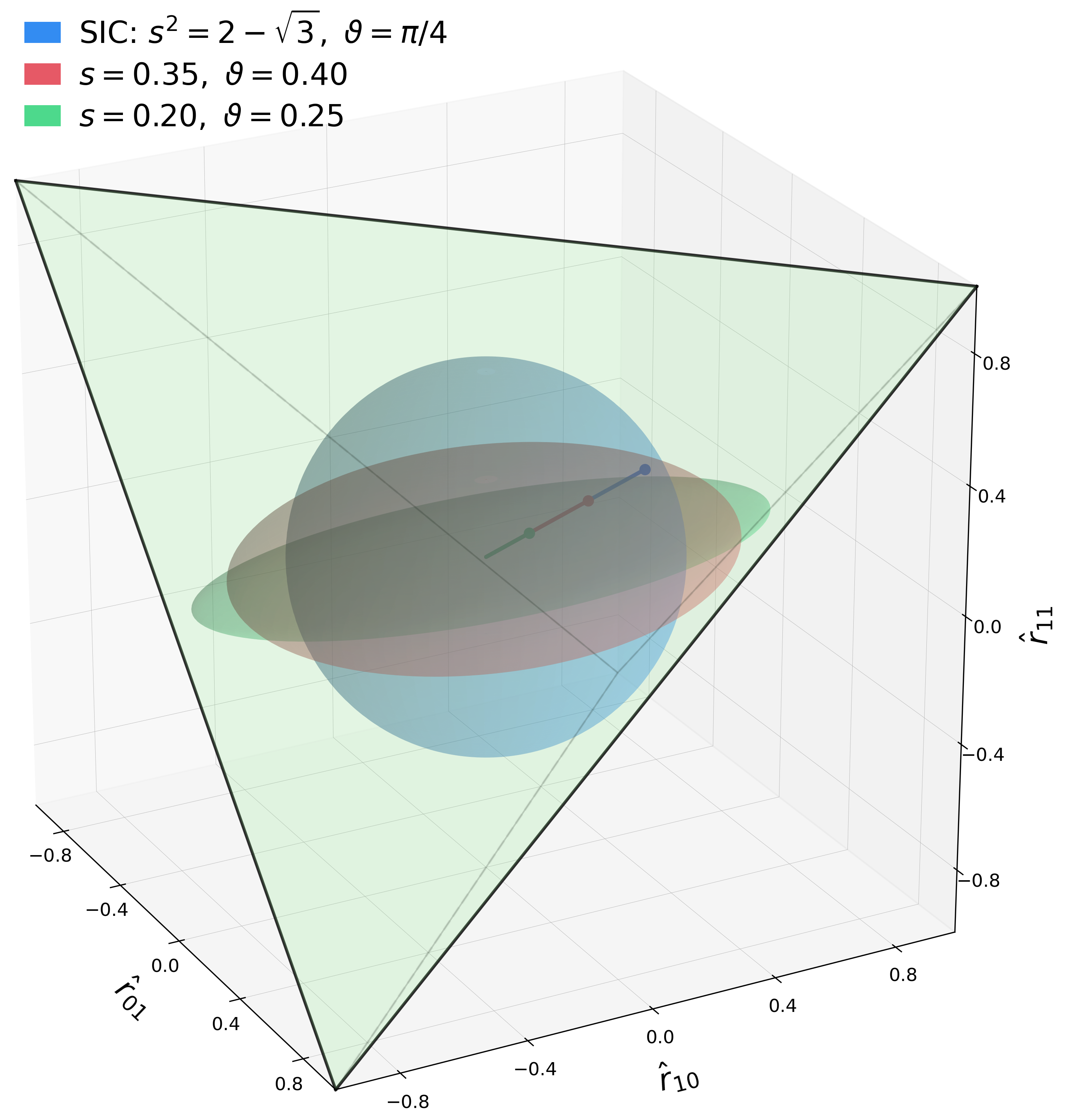}
  \caption{Compatibility regions $\mathscr C_\xi[\mathscr I^{\rm sc}_{\!4}]$ for three Heisenberg--Weyl extremals given in the Fourier coordinates $(\hat r_{01},\hat r_{10},\hat r_{11})$ of the shift-covariant observable $\P_{\mathbf r}\in\mathscr{I}^{\rm sc}_{\!4}$. The HW family is parameterised by $s\in(0,1)$ and $\vartheta\in(0,\pi/2)$, with coherence matrices given in Eq. \eqref{eq:hw-extremal-xi}.  The compatibility regions are centred ellipsoids; at the SIC point $s^2=2-\sqrt{3}$, $\vartheta=\pi/4$, the ellipsoid becomes the blue sphere. The red and green ellipsoids correspond to $(s,\vartheta)=(0.35,0.40)$ and $(0.20,0.25)$. The coloured line segments represent the depolarised-basis family $\P_\alpha\in\mathscr{I}^{\rm dep}_{\!4}$, lying along $(\hat r_{01},\hat r_{10},\hat r_{11})=\alpha(1,1,1)$. The tetrahedron is the set of all valid incoherent observables.}
  \label{fig:hw-ellipsoids}
\end{figure}

\begin{figure*}[t!]
    \centering
    \includegraphics[width=0.85\textwidth]{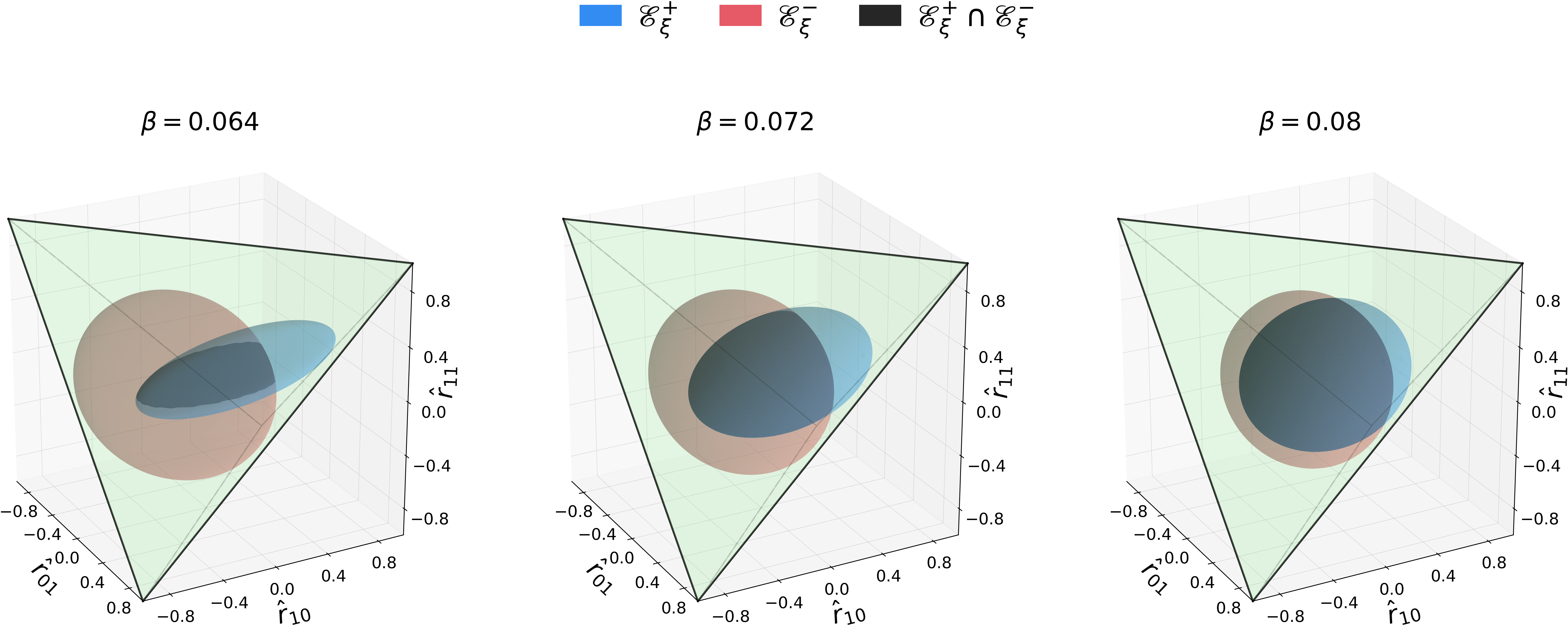}
\caption{Compatibility regions $\mathscr C_\xi[\mathscr I^{\rm sc}_{\!4}]$, restricted to shift-covariant observables, for three semi-SIC extremals, shown in the Fourier coordinates $(\hat r_{01},\hat r_{10},\hat r_{11})$ of the covariant seed distribution of $\P_{\mathbf r}$. The semi-SIC family is parametrised by $\beta\in(1/16,1/12]$, with coherence matrices given in Eq.~\eqref{eq:semi-sic-xi} and SIC endpoint at $\beta=1/12$. For $\beta<1/12$, the reduced symmetry splits the four outcomes into two inequivalent classes, so the compatibility region is the intersection of two ellipsoids $\mathscr E^\pm_\xi$ defined in Eq.~\eqref{app:eq:Epm-semisic}. The blue and red transparent surfaces are $\mathscr E^+_\xi$ and $\mathscr E^-_\xi$, and the dark region is $\mathscr E^+_\xi\cap\mathscr E^-_\xi=\mathscr C_\xi[\mathscr I^{\rm sc}_{\!4}]$. The tetrahedron is the set of all valid incoherent observables.}
    \label{fig:semisic-intersections}
\end{figure*}

We can now immediately find the compatibility volume relative to the SIC-extremal using Theorem    ~\ref{thm:volume-comparison-sic}:
\begin{equation}\label{HWMIC_vol}
\vol(\xi\,;\xi_{\rm SIC},\mathscr I^{(m)}_{\!4}) =\left[      \frac{108s^4(1-s^2)^2\sin^2(2\vartheta)}{(1+s^2)^6}\right]^{\frac{m-1}{2}}.
\end{equation}

To determine the detailed structure of $\mathscr C_\xi$, we evaluate the ellipsoid $\mathscr E_\xi$ from Proposition~\ref{ellipses}.  Since the HW MICs have $e^\xi_{kl}=1/2$ for all $(k,l)\in\Omega$, Eq.~\eqref{eq:x-coordinates} reduces to $x^\P_{nm}(j)=\hat p_{nm}(j)/\hat p_{00}(j)$ for $(n,m)\in\Omega_0$, where $\hat p_{nm}(j):=\sum_{(k,l)\in\Omega}(-1)^{kn+lm}p_{kl}(j)$ are the Hadamard--Fourier coefficients of the probability vector $\mathbf p(j)$ of $\P$ and $\Omega_0:=\Omega\setminus\{(0,0)\}$. Thus an incoherent observable $\P$ belongs to $\mathscr C_\xi$ if and only if, for each outcome $j$, the vector $\mathbf x^\P(j)$ lies in $\mathscr E_\xi$.

The SIC case is contained in this family at $s^2=2-\sqrt3$ and $\vartheta=\pi/4$, where $\lambda_{01}=\lambda_{10}=\lambda_{11}=\frac13$. The ellipsoid becomes a sphere only at this point, and the compatibility volume attains its maximum value $1$.

We now specify the compatibility region $\mathscr{C}_\xi[\mathscr I_4^{\rm sc}]$ for the class of shift-covariant probe observables. As shown in Appendix~\ref{app:hw-extremals}, in the Fourier-parametrisation \eqref{fouriercoords} one obtains
\begin{equation}\label{eq:HW-cov-region}
\mathscr C_\xi[\mathscr I^{\rm sc}_{\!4}]
=
\left\{
\P_{\mathbf r}\in\mathscr I^{\rm sc}_{\!4}
\;\middle|\;
\sum_{(n,m)\in\Omega_0}
\frac{\hat r_{nm}^{\,2}}{\lambda_{nm}}
\leq 1
\right\}.
\end{equation}
Thus the shift-covariant compatibility region is given by a \emph{single} ellipsoid.

Figure~\ref{fig:hw-ellipsoids} illustrates $\mathscr C_\xi[\mathscr I^{\rm sc}_{\!4}]$ for several
choices of the parameters $(s,\vartheta)$. A simple shift-covariant example is the depolarised-basis
observable $\P_\alpha\in\mathscr I_{\!4}^{\rm dep}$ defined in Eq.~\eqref{eq:line-observable}, corresponding to the seed distribution
\[
\mathbf r_\alpha
=
\frac14(1+3\alpha,1-\alpha,1-\alpha,1-\alpha).
\]
Its nontrivial Fourier coordinates are $(\hat r_{01},\hat r_{10},\hat r_{11})
=
\alpha(1,1,1)$, so the family appears as the diagonal line segment in Fig.~\ref{fig:hw-ellipsoids}. The exact compatibility threshold $\alpha^*(s,\vartheta)$ for Heisenberg--Weyl extremals is given in
Corollary~\ref{app:cor:dariano-app} of Appendix~\ref{app:hw-extremals}. At the SIC point one obtains $\alpha^*=1/3$, and the compatible portion of this depolarised family is maximal, in agreement with Proposition~\ref{prop:sic-maximal-line}.

\subsection{Semi-SIC extremals}
\label{sec:semi-sic-extremals}

We now consider a family of extremals whose associated MICs are known as \emph{semi-SICs}. These POVMs relax the symmetry of a SIC by allowing the effects to have different traces while retaining equal pairwise overlaps, that is for some constant $\beta$, we have ${\rm tr}[\M(n,m)\M(n',m')]=\beta$ when $(n,m)\neq (n',m')$. They were introduced and characterised in \cite[Theorem~1]{geng21}; for $r=2$ they form a one-parameter family where the overlap parameter has $\beta\in(1/16,1/12]$, the endpoint $\beta=1/12$ corresponding to a SIC.

The corresponding family of coherence matrices is
\begin{equation}\label{eq:semi-sic-xi}
\xi =
\begin{pmatrix}
1 & \gamma & \tfrac{1}{\sqrt{3}} & \tfrac{1}{\sqrt{3}} \\
\gamma & 1 &
  \frac{(\gamma-\nu\,e^{i\theta})}{\sqrt{3}} &
  \frac{(\gamma-\nu\,e^{-i\theta})}{\sqrt{3}}
  \\
\tfrac{1}{\sqrt{3}} &
  \frac{1(\gamma-\nu\,e^{-i\theta})}{\sqrt{3}} &
  1 & \tfrac{1+2e^{-2i\theta}}{3}
  \\
\tfrac{1}{\sqrt{3}} &
  \frac{(\gamma-\nu\,e^{i\theta})}{\sqrt{3}} &
  \tfrac{1+2e^{2i\theta}}{3} & 1
\end{pmatrix},
\end{equation}
where
\begin{align*}
\gamma&=\frac{2\sqrt{\beta}}{1-\sqrt{1-12\beta}}, &
\theta&=\cos^{-1}\!\left(
\frac{\sqrt{1-8\beta-\sqrt{1-12\beta}}}{4\sqrt{\beta}}
\right),\\
\nu &=\sqrt{2(1-\gamma^2)}.
\end{align*}
It follows that $e_{00}=e_{01}=e_+$ and $e_{10}=e_{11}=e_-$, with
\begin{equation}\label{eq:semi-sic-e}
e_\pm=\frac12\bigl(1\pm\sqrt{1-12\beta}\bigr),
\end{equation}
so the semi-SICs are not unbiased. More details are given in Appendix~\ref{app:semi-sic-proof}.

In particular, as shown in Appendix~\ref{app:semi-sic-proof}, the relevant matrix $H_\xi$ has a simple spectral structure. The eigenbasis can now be chosen as
\begin{equation}\label{eq:semisic-basis}
\begin{pmatrix}
\mathbf v_{00} & \mathbf v_{01} & \mathbf v_{10} & \mathbf v_{11}
\end{pmatrix}
=
\frac{1}{\sqrt 2}
\begin{pmatrix}
\sqrt{e_+} &  \sqrt{e_-} &  1 &  0 \\
\sqrt{e_+} & \sqrt{e_-} &  -1 & 0 \\
\sqrt{e_-} &  -\sqrt{e_+} & 0 & 1 \\
\sqrt{e_-} & -\sqrt{e_+} & 0 &  -1
\end{pmatrix},
\end{equation}
with the corresponding eigenvalues
\begin{align*}
\lambda_{00}&=1, & \lambda_{01}&=\frac13,&
\lambda_{10}&=a_+, &\lambda_{11}&=a_-,
\end{align*}
where $a_\pm =(1\pm2\sqrt{1-12\beta})/3$. At the SIC endpoint $\beta=1/12$, one has $a_+=a_-=1/3$, and therefore all three nontrivial eigenvalues coincide. Hence (interestingly) the semi-SIC nature is reflected in the last two eigenvalues only.

\begin{figure*}[t!]
    \centering
    \includegraphics[width=0.85\textwidth]{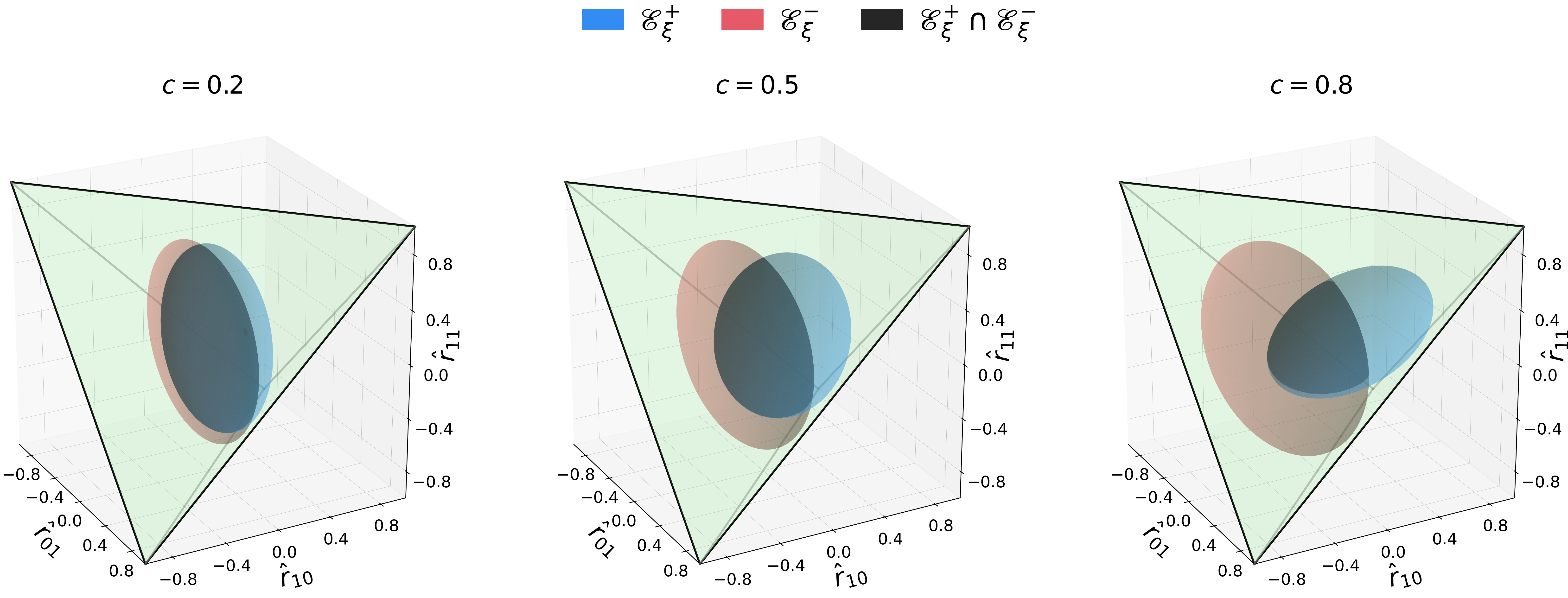}
\caption{Compatibility regions $\mathscr C_\xi[\mathscr I^{\rm sc}_{\!4}]$, restricted to shift-covariant observables, for three non-MIC extremals, shown in the Fourier coordinates $(\hat r_{01},\hat r_{10},\hat r_{11})$ of the covariant seed distribution of $\P_{\mathbf r}$. The family is parametrised by $c=\cos^2(\phi/2)\in(0,1)$, with coherence matrices given in Eq.~\eqref{eq:non-mic-xi}. The compatibility region is the intersection of two ellipsoids $\mathscr E^\pm_\xi$ defined in Eq.~\eqref{app:eq:Epm-nonmic}, arising from the two inequivalent outcome classes of the covariant observable. The blue and red transparent surfaces are $\mathscr E^+_\xi$ and $\mathscr E^-_\xi$, and the dark region is $\mathscr E^+_\xi\cap\mathscr E^-_\xi=\mathscr C_\xi[\mathscr I^{\rm sc}_{\!4}]$. The tetrahedron corresponds to the set of all valid incoherent observables.}
    \label{fig:nonmic-intersections}
\end{figure*}

The compatibility volume relative to the SIC-extremal now reads
\begin{equation}\label{SEMISIC_vol}
\vol(\xi\,;\xi_{\rm SIC},\mathscr I^{(m)}_{\!4}) =\left[ \frac{16\beta-1}{48\beta^2} \right]^{\frac{m-1}{2}}\,.
\end{equation}
To describe the full compatibility region $\mathscr C_\xi$, fix an outcome $j$ of $\P\in\mathscr I_{\!4}$ and let $\mathbf p=(p_{00},p_{01},p_{10},p_{11})^T$ be the corresponding coefficient vector. The coordinates of $\mathbf x_\xi^\P$ in Eq.~\eqref{eq:x-coordinates} are then
\begin{align}\label{eq:semi-sic-x}
x_{01} &=
\sqrt{3\beta}N^{-1}\,
(p_{00}+p_{01}-p_{10}-p_{11}),\\
x_{10} &=
\sqrt{e_+}N^{-1}\, (p_{00}-p_{01}),\,\, x_{11} =
\sqrt{e_-}\,N^{-1}(p_{10}-p_{11})\nonumber,
\end{align}
where $N=\sum_{n,m=0}^1 e_{nm} p_{nm}$. Proposition~\ref{ellipses} therefore implies that $\P\in\mathscr C_\xi$ iff, for each outcome $j$, the corresponding vector $\mathbf x^{\P}_\xi=(x_{01},x_{10},x_{11})$ lies in the ellipsoid
\[
\mathscr E_\xi
=
\left\{
\mathbf x\in\mathbb R^3
\;\middle|\;
3x_{01}^2
+\frac{x_{10}^2}{a_+}
+\frac{x_{11}^2}{a_-}
\leq 1
\right\}.
\]

We now again restrict to shift-covariant incoherent observables $\P_{\mathbf r}\in\mathscr I^{\rm sc}_{\!4}$ defined in Eq.~\eqref{eq:inc-sh} and characterise the region $\mathscr C_\xi[\mathscr I^{\rm sc}_{\!4}]\subseteq \mathscr C_\xi$ for semi-SIC extremals. Since the coordinates of $\mathbf e^\xi=(e_+,e_+,e_-,e_-)^T$ are generally non-uniform, the normalisation in Eq.~\eqref{eq:semi-sic-x} is not shift-invariant. The four outcomes therefore split into two inequivalent classes, and the covariant compatibility region is the intersection of two ellipsoids:
\begin{equation}\label{eq:semi-sic-hw}
\mathscr C_\xi[\mathscr I^{\rm sc}_{\!4}]
=
\left\{
\P_{\mathbf r}\in\mathscr I^{\rm sc}_{\!4}
\;\middle|\;
\hat{\mathbf r}\in\mathscr E^+_\xi\cap\mathscr E^-_\xi
\right\}\,,
\end{equation}
where $\hat{\mathbf r}$ again denotes the Fourier-coordinate vector \eqref{fouriercoords} of the seed $\mathbf r$. The explicit forms of $\mathscr E^\pm_\xi$ and the proof of Eq.~\eqref{eq:semi-sic-hw} are given in Corollary~\ref{app:cor:semi-sics-app} of Appendix~\ref{app:semi-sic-proof}.

Figure~\ref{fig:semisic-intersections} illustrates these intersections for several choices of $\beta$. The SIC endpoint is recovered when the two ellipsoids coincide and become the sphere of the preceding Heisenberg--Weyl SIC case.

\subsection{A non-MIC family}
\label{sec:nonmic-family}

Finally, we demonstrate that our framework is more general that the MIC-case; we consider a one-parameter family of extremals in $\mathcal E_\geq (2)$ that do \emph{not} belong to $\mathcal E_{>}(2)$. Hence, even though they define POVMs in the dilation space, they are not associated with MICs. In this case Proposition~\ref{ellipses} does not apply, and we instead use the criterion of Proposition~\ref{prop:qubit-K-from-G}.

The coherence matrices in question are given by
\begin{equation}\label{eq:non-mic-xi}
\xi=
\begin{pmatrix}
1 & 0 & \frac{1}{\sqrt2} & \frac{1}{\sqrt2}\\
0 & 1 & \frac{1}{\sqrt2} & \frac{e^{i\phi}}{\sqrt2}\\
\frac{1}{\sqrt2} & \frac{1}{\sqrt2} & 1 & \frac{1+e^{i\phi}}{2}\\
\frac{1}{\sqrt2} & \frac{e^{-i\phi}}{\sqrt2} & \frac{1+e^{-i\phi}}{2} & 1
\end{pmatrix},
\end{equation}
where $\phi\in(0,\pi)$. This family is obtained from the example in \cite{buscemi05}, corresponding to $\phi=\pi/2$, by introducing a relative phase $e^{i\phi}$ in one of the structure vectors, see Appendix~\ref{app:nonmic-family}.

For these extremals, one has $\mathbf e^\xi=(1,1,0,0)^T$, hence $\xi\in\Eobs$ but $\xi\notin\Ep$. Thus the decoherence channel does not correspond to an informationally complete POVM; instead, $\M=\Phi(\xi)$ is simply the projective basis measurement $\M(0,0)=|0\rangle\langle 0|$, $\M(0,1)=|1\rangle\langle 1|$, with the other two effects zero. However, Theorem \ref{thm:volume-comparison-sic} still applies, and we obtain the compatibility volume relative to the SIC-extremal as
\begin{equation}\label{non_MIC_vol}
\vol(\xi\,;\xi_{\rm SIC},\mathscr I^{(m)}_{\!4}) =\left[\frac 14 \sin^2\phi\right]^{\frac{m-1}{2}},
\end{equation}
which now remains strictly below $1$ as the SIC-case is not included.

In order to characterise the compatibility region $\mathscr C_\xi$, we first fix an outcome $j$, let $\mathbf p(j)=(p_{00},p_{01},p_{10},p_{11})^T$, and set $h:=\frac12(p_{00}+p_{01})$. By applying Proposition~\ref{prop:qubit-K-from-G} (see Appendix~\ref{app:nonmic-family}), the compatibility condition is equivalent to
\begin{equation}\label{eq:non-mic-criterion}
(p_{10}-h)^2+(p_{11}-h)^2-(2c-1)(p_{10}-h)(p_{11}-h)\le c\,p_{00}p_{01},
\end{equation}
where we used the convenient parameter $c=\cos^2(\phi/2)\in (0,1)$. Thus, $\P\in\mathscr C_\xi$ if and only if Eq. \eqref{eq:non-mic-criterion} holds for every outcome $j$.

To compare directly with the preceding two extremal families, we again restrict to shift-covariant incoherent observables $\P_{\mathbf r}\in\mathscr I^{\rm sc}_{\!4}$ and characterise $\mathscr C_\xi[\mathscr I^{\rm sc}_{\!4}]$. Using the Fourier-coordinates \eqref{fouriercoords} we obtain from \eqref{eq:non-mic-criterion} the ellipsoid
\begin{equation}\label{eq:non-mic-ellipsoid}
\begin{aligned}
&c\left(\hat r_{01}+\hat r_{11}\right)^2
+\frac{c}{1-c}\left(\hat r_{01}-\hat r_{11}\right)^2 \\
&\quad +(4-c)\left(\hat r_{10}-\frac{c}{4-c}\right)^2
\le \frac{4c}{4-c}\,,
\end{aligned}
\end{equation}
as shown in Appendix~\ref{app:nonmic-family}.

Thus, the condition for the seed effect is an ellipsoid centred at $(0,\frac{c}{4-c},0)$ in the coordinates $(\hat r_{01},\hat r_{10},\hat r_{11})$. Its principal directions in the $(\hat r_{01},\hat r_{11})$-plane are the diagonals $\hat r_{01}=\pm\hat r_{11}$, and the ratio of the corresponding semi-axis lengths is $\sqrt{1-c}$. Hence this section becomes increasingly elongated as $c$ grows, while the centre shifts monotonically in the positive $\hat r_{10}$-direction.

Under translation by $(n,m)$, the nontrivial Fourier coefficients transform as
\[
(\hat r_{01},\hat r_{10},\hat r_{11})\mapsto
\bigl((-1)^m \hat r_{01},\;(-1)^n \hat r_{10},\;(-1)^{n+m}\hat r_{11}\bigr).
\]
It follows that the four outcomes split into two inequivalent classes, namely $(0,0)\sim(0,1)$ and $(1,0)\sim(1,1)$. Therefore, as in the semi-SIC case, the covariant compatibility region is the intersection of two ellipsoids:
\begin{equation}\label{eq:non-mic-hw}
\mathscr C_\xi[\mathscr I^{\rm sc}_{\!4}]
=
\left\{
\P_{\mathbf r}\in\mathscr I^{\rm sc}_{\!4}
\;\middle|\;
\hat{\mathbf r}\in\mathscr E^+_\xi\cap\mathscr E^-_\xi
\right\}.
\end{equation}
The ellipsoids $\mathscr E^\pm_\xi$ are given explicitly in Corollary~\ref{app:cor:nonmic-hw-app} of Appendix~\ref{app:nonmic-family}; they are obtained from Eq.~\eqref{eq:non-mic-ellipsoid} by applying the above sign changes to the translated seed effects. We illustrate the compatibility region $\mathscr C_\xi[\mathscr I^{\rm sc}_{\!4}]$ for several values of $\phi$ in Fig.~\ref{fig:nonmic-intersections}.

To conclude the results of this section, we have characterised the shift-covariant region $\mathscr C_\xi[\mathscr I^{\rm sc}_{\!4}]$ for three distinct families: Heisenberg--Weyl MICs, semi-SICs, and certain non-MIC extremals. Together, these examples illustrate the different compatibility geometries induced by extremal decoherence channels. For the Heisenberg--Weyl MIC family (Fig.~\ref{fig:hw-ellipsoids}), the shift-covariant compatibility region is a single centred ellipsoid in Fourier space, reducing to a sphere at the SIC point. For the semi-SIC and non-MIC families (Figs.~\ref{fig:semisic-intersections} and \ref{fig:nonmic-intersections}, respectively), the reduced symmetry breaks the equivalence of the four outcomes and leads instead to the intersection of two ellipsoids. In the semi-SIC case the ellipsoids are shifted in opposite directions along the $\hat r_{10}$-axis; as the SIC point is approached, the shifts vanish and the two ellipsoids merge into the SIC sphere. In the non-MIC case the two ellipsoids are also shifted away from the origin, but with a different orientation and eccentricity, producing a distinct intersection geometry. The compatibility volumes of each family, relative to an extremal-SIC channel, are shown in Fig. \ref{fig:volumes} for the full set of four-outcome incoherent observables.

\begin{figure}
\includegraphics[width=0.48\textwidth]{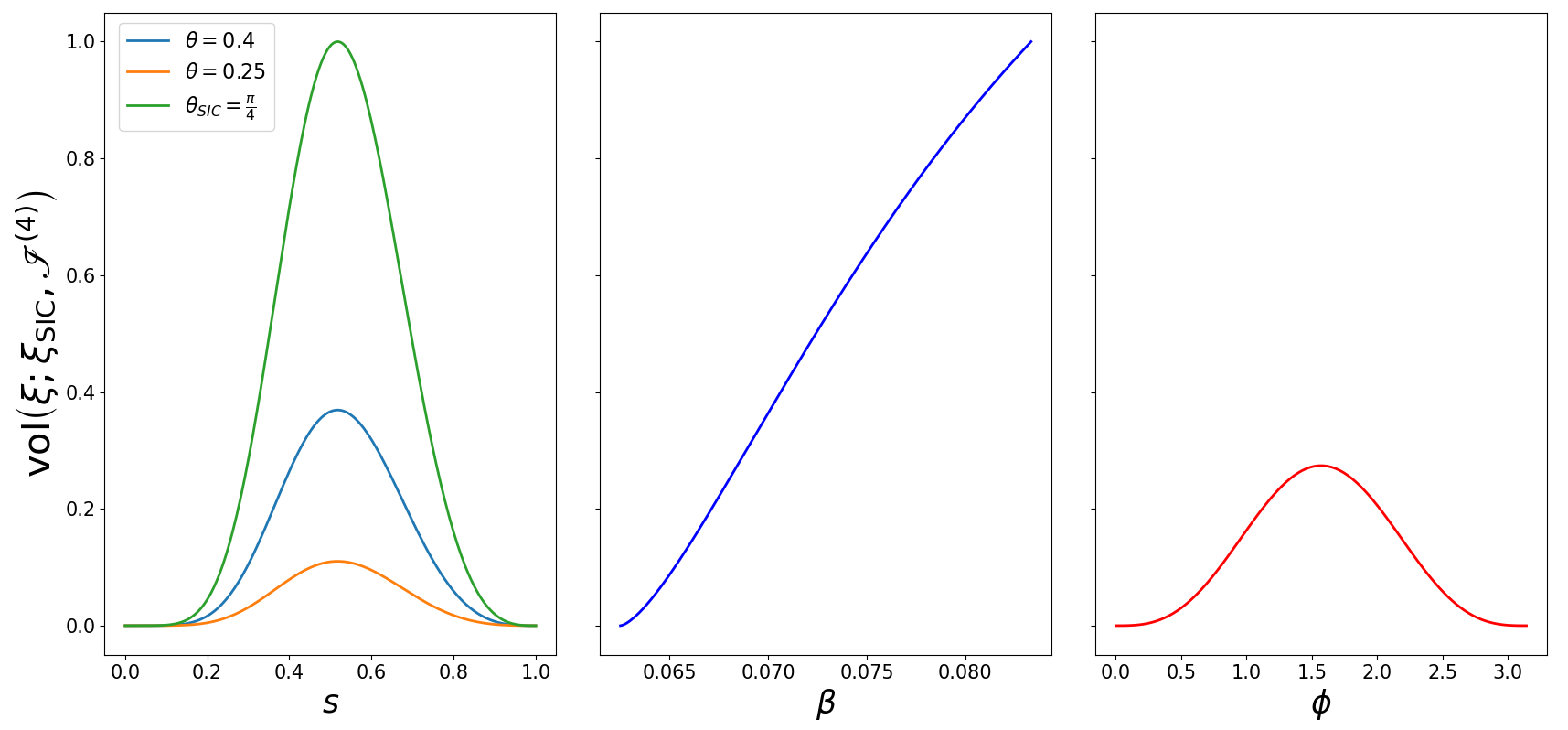}
\caption{The compatibility volume of $\xi$ relative to a SIC-extremal, namely $\vol(\xi\,;\xi_{\rm SIC},\mathscr I^{(4)}_{\!4})$, for the HW family (left), the semi-SIC family (centre), and the non-MIC family (right), as functions of the relevant parameters. The volume is computed for the restricted compatibility region $\mathscr C_\xi[\mathscr I^{(4)}_{\!4}]$ of four-outcome incoherent observables, with explicit formulae given in Eqs. \eqref{HWMIC_vol}, \eqref{SEMISIC_vol} and \eqref{non_MIC_vol}. For the HW and semi-SIC families, the maximum value $1$ is reached exactly at the parameter values corresponding to SIC-extremals.}
    \label{fig:volumes}
\end{figure}

\section{Summary and outlook} \label{sec:conclusion}

In this work we have studied how quantum noise in the form of decoherence causes loss of measurement incompatibility. In particular, we characterised the compatibility region $\mathscr{C}_\xi$ induced by the channel, namely the set of incoherent observables that become jointly measurable with every observable after Schur multiplication by $\xi$. 

Our focus has been on extremal channels, where membership in $\mathscr C_\xi$ simplifies from a semidefinite feasibility problem to a direct positivity check of an explicit family of operators on the dilation space. This simplification comes from the rank-one operator frame associated with the extremal channel, and leads to connections with QBism and MIC state spaces. We introduced a scalar quantifier of the incompatibility loss through the volume of $\mathscr C_\xi$, which for maximal-rank extremals is controlled by the Gram matrix of the associated rank-one operator frame. This gives a quantitative way to compare decoherence effects, and distinguishes SIC-extremals as those that destroy the most incompatibility. As a consequence, the SIC existence problem can be reformulated as a joint measurability problem. We have illustrated the compatibility regions for several families of channels when restricted to covariant incoherent observables. These simplify to single ellipsoids (Fig.~\ref{fig:hw-ellipsoids}) or the intersection of ellipsoids (Figs.~\ref{fig:semisic-intersections}--\ref{fig:nonmic-intersections}). A comparison of their volumes can be found in Fig.~\ref{fig:volumes}.

A central conceptual finding is that incompatibility loss is not determined by the damping rates alone. Extremal and non-extremal channels with identical off-diagonal damping amplitudes can produce markedly different compatibility regions, particularly as the system size grows. Thus, the phase geometry encoded by the coherence matrix and reflected in its associated operator frame---not merely the magnitudes of its entrywise damping factors---is operationally significant for measurement incompatibility.

The framework developed here fits naturally within the broader theory of dynamical quantum resources \cite{saxena20}. In particular, it supplements the resource theory of incompatibility preservability \cite{hsieh25} by providing an analytically tractable geometric characterisation of \emph{how} incompatibility is preserved. The region $\mathscr P_\xi=\mathscr I_{\!d}\setminus\mathscr C_\xi$ identifies exactly those incoherent observables whose incompatibility survives decoherence, while its relative volume quantifies how prevalent this preservability is. This approach can distinguishes channels that are indistinguishable by scalar robustness quantifiers, thereby providing a finer operational characterisation of incompatibility preservability under decoherence.

Beyond these connections, there are several natural directions for future research. While extremality has previously appeared in the context of joint measurements \cite{haapasalo14,guerini18,carmeli19b}, our work highlights a complementary role for extremality on the channel side. It would be interesting to understand how far this simplifying role extends beyond decoherence, and whether similar mechanisms arise for other classes of noisy dynamics. Another direction is to go beyond the maximal-rank extremal case considered here and investigate how the rank of the coherence matrix affects the loss of incompatibility. Time-dependent decoherence provides a further natural setting: for suitably divisible dynamics, the compatibility region expands monotonically, so departures from monotonicity could be investigated as a signature of non-Markovianity. Similar questions could also be pursued for other measurement-based nonclassical resources.

\section*{Acknowledgements}

D.M. acknowledges support from PNRR MUR Project No. PE0000023-NQSTI and from INFN through the project “QUANTUM”. W.T. acknowledges support from Aberystwyth University AberDoc Scholarship.

\bibliographystyle{apsrev4-2}
\bibliography{extremals_refs}

\appendix

\section{Proof of Prop.~\ref{prop:equiv}}\label{secA:prop-equiv}

\renewcommand{\theproposition}{\ref{prop:equiv}}
\begin{proposition}
The restriction of $\Phi$ to $\Ep$ is bijective onto $\MIC$, and satisfies
\begin{equation}\label{obsGIO}
{\rm tr}[\M(n)\M(m)]=e^\xi_n e^{\xi}_m \,(G_\xi)_{nm}
\end{equation}
for all $\xi\in \Ep$, $\M=\Phi(\xi)$, and $n,m\in \{1,\ldots, r^2\}$.
\end{proposition}

\begin{proof}
Let $\xi\in \Ep$ and $\M=\Phi(\xi)$. Since $\xi$ is extremal and ${\mathbf e}^\xi>0$, we have ${\rm span}\, \{\M(n)\}={\rm span}\, \{e^\xi_n|a_n\rangle\langle a_n|\}={\rm span}\, \{|a_n\rangle\langle a_n|\} =\mathcal{B(K)}$, so $\M$ is informationally complete. Hence $\M\in \MIC$.

In order to prove injectivity, let $\xi,\xi'\in \Ep$, and choose unit vectors $a_n$, $a_n'$ for the (representative) matrices $\xi,\xi'$, respectively. Now if $\Phi(\xi) = \Phi(\xi')$, then there is a unitary $U$ such that $e_n^\xi U|a_n\rangle\langle a_n|U^* =e_n^{\xi'} |a_n'\rangle\langle a_n'|$ for all $n$. By taking the trace we find ${\mathbf e}^\xi={\mathbf e}^{\xi'}>0$, so $U|a_n\rangle\langle a_n|U^* =|a_n'\rangle\langle a_n'|$ for all $n$, which implies the existence of phase factors $v_n$ such that $Ua_n = v_n a'_n$ for all $n$. Hence $\xi\sim \xi'$. This shows that $\Phi$ is injective when restricted to $\Ep$.

To prove surjectivity, take an $\M\in \MIC$. Since each $\M(n)$ has rank 1, we can pick a unit vector $a_n\in {\rm ran}\, \M(n)$ for each $n\in \{1,\ldots, r^2\}$, to write $\M(n) = {\rm tr}[\M(n)] |a_n\rangle\langle a_n|$, where ${\rm tr}[\M(n)] >0$. We then define a coherence matrix $\xi$ by $\xi_{nm}=\langle a_n|a_m\rangle$. Since $\M$ is informationally complete and ${\rm tr}[\M(n)] >0$ for each $n$, we conclude that $\xi$ is extremal, and by the normalisation of the POVM $\M$ we get $\id_r = \sum_n \M(n) = \sum_n {\rm tr}[\M(n)] |a_n\rangle\langle a_n|$, which by the uniqueness of coefficients in \eqref{compl} shows that $e_n^\xi={\rm tr}[\M(n)]>0$ for all $n$. Hence $\xi\in \Ep$, and since $\M(n) = e^\xi_n |a_n\rangle\langle a_n|$ we have $\M = \Phi(\xi)$. This shows that $\Phi$ restricted to $\Ep$ is surjective onto $\MIC$, and the property \eqref{obsGIO} follows as $\M(n) = e^\xi_n |a_n\rangle\langle a_n|$ for all $n$.
\end{proof}

\section{Proof of Prop.~\ref{prop:sic-maximal-line}}
\label{secA:SIC-line-equiv}

The following result gives the sharp compatibility threshold for the depolarised incoherent basis observable $\P_\alpha$ in the restricted compatibility region $\mathscr C_\xi[\mathscr I^{\rm dep}_{\!d}]$ for an extremal $\xi\in\ext(r)$.

\renewcommand{\theproposition}{\ref{prop:sic-maximal-line}}
\begin{proposition}
Let $\xi\in\ext(r)$. Then
\begin{equation}\label{eq:sic-line-bound}
\alpha^*(\xi)\le\frac1{r+1}\,,
\end{equation}
with equality if and only if $\xi$ is SIC-extremal.
\end{proposition}

\begin{proof}
Let $\xi_{mn}=\ip{a_m}{a_n}$, and write $\Pi_n^\xi=|a_n\rangle\langle a_n|$ for the structure projectors on $\mathcal K\simeq\mathbb C^r$. Suppose $\P_\alpha\in\mathscr C_\xi[\mathscr I^{\rm dep}_{\!d}]$. By Theorem~\ref{thm:gii}, there exists a POVM $\F$ on $\mathcal K$ such that $\tr{\F(j)\Pi_n^\xi}=p_\alpha(j|n)$ for all $j,n=1,\ldots,r^2$. Since $\Pi^\xi_j\le\id_{\mathcal K}$ and $\F(j)\ge0$, we have
\[
\tr{\F(j)}\ge\tr{\F(j)\Pi^\xi_j}
=
\alpha+\frac{1-\alpha}{r^2}.
\]
Summing over $j$ gives
\[
r=\sum_{j=1}^{r^2}\tr{\F(j)}\ge r^2\left(\alpha+\frac{1-\alpha}{r^2}\right)=1+(r^2-1)\alpha .
\]
Hence $\alpha\le 1/(r+1)$, and therefore $\alpha^*(\xi)\le 1/(r+1)$.

Suppose $\alpha^*(\xi)=1/(r+1)$. Since $\mathscr C_\xi[\mathscr I^{\rm dep}_{\!d}]$ is closed, $\P_{1/(r+1)}\in\mathscr C_\xi[\mathscr I^{\rm dep}_{\!d}]$. For this endpoint the trace inequality above is saturated for every $j$, and hence
\[
\tr{\F(j)}=\tr{\F(j)\Pi^\xi_j}=\frac1r.
\]
Thus $\F(j)$ has no support outside $\operatorname{ran}\Pi^\xi_j$. Since $\Pi^\xi_j$ is rank one, $\F(j)=\lambda_j\Pi^\xi_j$, and taking the trace gives $\lambda_j=1/r$. Hence $\F(j)=\Pi^\xi_j/r$. For $n\neq j$,
\[
\frac{1}{r(r+1)}
=p_{1/(r+1)}(j|n)=\tr{\F(j)\Pi_n^\xi}=\frac1r\tr{\Pi^\xi_j\Pi_n^\xi}.
\]
Therefore $\tr{\Pi^\xi_j\Pi_n^\xi}=1/(r+1)$ for all $j\neq n$. Since $\sum_j\F(j)=\id$, we also have $\sum_j\Pi^\xi_j/r=\id$. Hence $\{\Pi^\xi_j/r\}_{j=1}^{r^2}$ is a SIC-POVM, and $\xi$ is SIC-extremal.

Conversely, if $\xi$ is SIC-extremal, then $\F(j)=\Pi^\xi_j/r$ is a POVM and satisfies $\tr{\F(j)\Pi_n^\xi}=p_{1/(r+1)}(j|n)$. Hence
$\P_{1/(r+1)}\in\mathscr C_\xi[\mathscr I^{\rm dep}_{\!d}]$, and the bound above gives $\alpha^*(\xi)=1/(r+1)$.
\end{proof}

\section{Proof of Lemma~\ref{lem:det-bound}}\label{app:det-bound}

We now provide an essential result leading to Theorem~\ref{thm:volume-comparison-sic}, in which we show that SIC-extremals are distinguished as those that maximise the compatibility region $\mathscr{C}_\xi$ among all extremals $\xi\in\ext(r)$. The following result is a reformulation, in the language of extremal coherence matrices, of a bound found in \cite{debrota20} (see also \cite{kwapisz19}).

\renewcommand{\thelemma}{\ref{lem:det-bound}}
\begin{lemma}
Let $\xi\in\ext(r)$. Then
\[
\det G_\xi
\le
r\left(\frac{r}{r+1}\right)^{r^2-1}\,,
\]
with equality if and only if $\xi$ is SIC-extremal.
\end{lemma}
\begin{proof}
Let $\xi_{mn}=\ip{a_m}{a_n}$ and $\Pi_n^\xi=|a_n\rangle\langle a_n|$. Since $(G_\xi)_{nm}
=|\langle a_n|a_m\rangle|^2=\tr{\Pi_n^\xi\Pi_m^\xi}$, the matrix $G_\xi$ is the Hilbert--Schmidt Gram matrix of the rank-one projectors $\{\Pi_n^\xi\}_{n=1}^{r^2}$. Since $\xi\in\ext(r)$, this Gram
matrix is positive definite by Proposition~\ref{Dinv}. Let $\lambda_1\ge\cdots\ge\lambda_{r^2}>0$ be the eigenvalues of $G_\xi$. By \cite[Lemma~3]{debrota20}, applied to the normalised positive semidefinite operator basis $\{\Pi_n^\xi\}_{n=1}^{r^2}$, one has $\lambda_1\ge r$, with equality in the resulting SIC spectrum if and only if the projectors form a SIC. 

Furthermore, since $\sum_{k=1}^{r^2}\lambda_k=\tr{G_\xi}=r^2$ we have $\sum_{k>1}\lambda_k=r^2-\lambda_1$. By the arithmetic--geometric mean inequality,
\[
\prod_{k>1}\lambda_k
\le
\left(\frac{r^2-\lambda_1}{r^2-1}\right)^{r^2-1}.
\]
Therefore
\[
\det G_\xi
\le
\lambda_1
\left(\frac{r^2-\lambda_1}{r^2-1}\right)^{r^2-1}.
\]
The right-hand side is decreasing for $\lambda_1\ge r$, and is therefore maximised at $\lambda_1=r$. This gives
\[
\det G_\xi
\le
r\left(\frac{r^2-r}{r^2-1}\right)^{r^2-1}
=
r\left(\frac{r}{r+1}\right)^{r^2-1}.
\]
Equality requires $\lambda_1=r$ and $\lambda_2=\cdots=\lambda_{r^2}=r/(r+1)$, which is exactly the equality case of \cite[Lemma~3]{debrota20}. Hence equality holds if and only if $\{r^{-1}\Pi_n^\xi\}_{n=1}^{r^2}$ forms a SIC.
\end{proof}

\section{Incompatibility preservability}
\label{app:resource-preservability}

In the resource theory of \emph{incompatibility preservability} introduced in \cite{hsieh25}, the resourcefulness of a channel is quantified by the robustness
\begin{equation}
\label{eq:IP-robustness}
R_{\rm IP}(\Lambda)
:=
\min\left\{
s\geq0\;\middle|\;
\frac{\Lambda+s\Phi}{1+s}\in\mathrm{IB},
\quad
\Phi\ \text{a channel}
\right\},
\end{equation}
where $\mathrm{IB}$ denotes the set of incompatibility-breaking channels and $\Phi$ is an arbitrary auxiliary channel on the same system. The optimisation ranges over all auxiliary channels and imposes the requirement that the resulting mixture be incompatibility-breaking. 

We now consider a simpler version, adapted to our decoherence setting, in which the auxiliary channel is required to be a decoherence channel. Since the completely dephasing channel is the unique incompatibility-breaking channel among decoherence channels, we define the \emph{Schur robustness} by
\begin{equation}
\label{eq:Schur-robustness-def}
\mathcal R^{\rm Schur}_{\rm IP}(\xi)
:=
\min\left\{
t\geq0\,\middle|\,
\frac{\xi+t\zeta}{1+t}=\mathbb I
\text{ for some }\zeta\in\mathfrak C_d
\right\},
\end{equation}
where $\mathbb I$ is the coherence matrix of the completely dephasing channel. This restricted robustness need not coincide with $R_{\rm IP}$, but it is explicitly computable in terms of the maximum eigenvalue $\lambda_{\max}(\xi)$ of the coherence matrix, and provides a useful comparison with the compatibility volume.

\begin{proposition}
\label{prop:Schur-robustness}
For every $\xi\in\mathfrak C_d$,
\begin{equation}
\label{eq:Schur-robustness}
\mathcal R^{\rm Schur}_{\rm IP}(\xi)
=
\lambda_{\max}(\xi)-1.
\end{equation}
If $\xi\in\ext(r)$, then
\begin{equation}
\label{eq:Schur-robustness-bound}
\mathcal R^{\rm Schur}_{\rm IP}(\xi)\geq r-1,
\end{equation}
with equality if and only if
$\{\Pi_n^\xi/r\}_{n=1}^{r^2}$ is an unbiased rank-one MIC.
\end{proposition}

\begin{proof}
If $\xi=\mathbb I$ then Eq.~\eqref{eq:Schur-robustness} is immediate. Let $t>0$ such that
\[
\zeta
=
\frac{(1+t)\mathbb I-\xi}{t}.
\]
This matrix has unit diagonal and is positive semidefinite if and only if $(1+t)\mathbb I-\xi\geq0$, hence $t\geq\lambda_{\max}(\xi)-1$.

Now suppose that $\xi\in\ext(r)$. Since $\rank\xi=r$, the matrix $\xi$ has $r$ nonzero eigenvalues whose sum is $\tr\xi=r^2$. Their average is therefore $r$, and hence $\lambda_{\max}(\xi)\geq r$. Combining this with Eq.~\eqref{eq:Schur-robustness} gives $\mathcal R^{\rm Schur}_{\rm IP}(\xi)\geq r-1$. Equality holds if and only if all $r$ nonzero eigenvalues of $\xi$ equal $r$.

Let $A:\mathbb C^{r^2}\to\mathcal K$ be the operator defined by $A|n\rangle=|a_n\rangle$. Then
\[
\xi=A^*A,
\qquad
S_\xi:=AA^*
=\sum_{n=1}^{r^2}\Pi_n^\xi,
\]
where $S_\xi$ is the frame operator. The operators $A^*A$ and $AA^*$ have the same nonzero eigenvalues \cite{casazza13}. Consequently, equality in Eq.~\eqref{eq:Schur-robustness} holds if and only if every eigenvalue of $S_\xi$ equals $r$, or equivalently $S_\xi=r\id_r$. In that case $\sum_{n=1}^{r^2}\Pi_n^\xi/r=\id_r$, so $\{\Pi_n^\xi/r\}_{n=1}^{r^2}$ is a rank-one POVM with effects of equal trace. Extremality implies that the projections $\{\Pi_n^\xi\}_{n=1}^{r^2}$ are linearly independent, and hence this POVM is informationally complete. It is therefore an unbiased rank-one MIC. The converse follows immediately.
\end{proof}

Thus, within $\ext(r)$, SIC-extremals minimise the robustness \eqref{eq:Schur-robustness}, but not uniquely. Every extremal associated with an unbiased rank-one MIC has the same value $r-1$. Therefore, both the robustness and compatibility volume place SIC-extremals among the least incompatibility-preserving extremals, but only the compatibility volume singles them out. Indeed, $\mathcal R^{\rm Schur}_{\rm IP}$ depends only on $\lambda_{\max}(\xi)$, whereas the compatibility volume depends on $G_\xi$ and is uniquely maximised by SIC-extremals. The qubit Heisenberg--Weyl MICs considered in Sec.~\ref{sec:hw-extremals} make this distinction clear. All members of the family satisfy $\mathcal R^{\rm Schur}_{\rm IP}=1$, while their relative compatibility volumes range from one to values arbitrarily close to zero, as illustrated in Fig.~\ref{fig:volumes}.

\section{Joint measurability criteria for extremals in $\mathcal E(2)$}\label{app:r=2}

Restricting to extremals $\xi\in\mathcal E_{>}(2)$, we provide a simple condition for joint measurability in terms of ellipsoids.

\renewcommand{\theproposition}{\ref{ellipses}}
\begin{proposition}
Let $\xi\in \mathcal E_{>}(2)$, and let $\lambda_n$ and $\mathbf u_n$ be as above. Let $\P\in \mathscr I_4$ with coefficient vectors $\mathbf p(j)$. Define the vectors $\mathbf x_\xi^\P(j)\in \mathbb R^{3}$ by the coordinates
\begin{equation}\label{app:eq:x-coordinates}
[\mathbf x_\xi^\P(j)]_n:=\sqrt{2}\,
\frac{\mathbf u_n^T\mathbf p(j)}
{[\mathbf e^\xi]^T\mathbf p(j)},
\qquad n=1,2,3.
\end{equation}
Then
\[
\mathscr C_\xi
=
\left\{
\P
\;\middle|\;
\mathbf x_{\xi}^\P(j)\in\mathscr E_\xi
\text{ for each }j
\right\}\,,
\]
where
\[
\mathscr E_\xi
:=
\left\{
\mathbf x\in \mathbb R^{3}
\ \middle|\
\sum_{n=1}^{3}\frac{x_n^2}{\lambda_n}\leq 1
\right\}.
\]

\end{proposition}

\begin{proof}
For a given diagonal observable $\P$, fix an outcome $j$ and let $\mathbf p=\mathbf p(j)$. Define the corresponding $\mathbf q = G_\xi^{-1}\mathbf p$, and $F=\sum_n q_n \Pi_n^\xi$, where $\Pi_n^\xi$ are the rank-one projectors defining $\xi$. We need to determine when $F\geq 0$ under the given assumptions on $\P$. Noting that $G_\xi^{-1}\mathbf 1 = \mathbf e^\xi$, the decomposition \eqref{xiabsdecomp} gives
\begin{align*}
\tr{F}
&= \sum_{n} q_n
= \mathbf 1^T G_\xi^{-1} \mathbf p
= [\mathbf e^\xi]^T\mathbf p>0,\\
\tr{F^2}
&=\sum_{n,m}q_nq_m (G_\xi)_{nm}
= \mathbf q^T G_\xi \mathbf q\\
&= \mathbf p^T G_\xi^{-1}\mathbf p
= \frac{\bigl([\mathbf e^\xi]^T\mathbf p\bigr)^2}{2}
+ \sum_{n=1}^{3}\frac{\bigl(\mathbf u_n^T\mathbf p\bigr)^2}{\lambda_n}.
\end{align*}
Therefore, by the above lemma, $F\geq 0$ iff
\[
1\geq \frac{\tr{F^2}}{\tr{F}^2}
=\frac{1}{2}\left(1 + \sum_{n=1}^{3}\frac{\bigl[\mathbf x_\xi^\P\bigr]_n^2}{\lambda_n}\right),
\]
where $[\mathbf x_\xi^\P]_n=\sqrt{2}\frac{\mathbf u_n^T\mathbf p}{[\mathbf e^\xi]^T\mathbf p}$. This is equivalent to $\sum_{n=1}^{3}\frac{[\mathbf x_\xi^\P]_n^2}{\lambda_n}\leq 1$.
\end{proof}

\section{Special extremals in $\mathcal E(2)$}\label{app:examples}

We now provide details of our characterisation of the compatibility regions for the three families of extremal coherence matrices $\xi\in\ext(2)$ desribed in Sec.~\ref{sec:qubit-mics}: (1) Heisenberg-Weyl extremals; (2) semi-SIC extremals; and (3) non-MIC extremals.

\subsection{Heisenberg--Weyl extremals}
\label{app:hw-extremals}

We recall the construction of the qubit Heisenberg--Weyl MICs used in Sec.~\ref{sec:hw-extremals}. They are generated from the seed state
\begin{equation}
\ket{\varphi}
=
\frac{1}{\sqrt{1+s^2}}
\bigl(\ket{0}+s e^{i\vartheta}\ket{1}\bigr),
\end{equation}
with $0<s<1$ and $0<\vartheta<\pi/2$, and the unitaries
\begin{equation}
U_{nm}
=
|0\rangle\langle m|+(-1)^n|1\rangle\langle m+1|,
\quad n,m\in\{0,1\},
\end{equation}
where the second index is understood modulo $2$. The MIC is
\[
\M(n,m)=\frac12\,|a_{nm}\rangle\langle a_{nm}|,
\qquad
\ket{a_{nm}}=U_{nm}\ket{\varphi}\,,
\]
and its corresponding extremal $\xi$ is given in Eq. \eqref{eq:hw-extremal-xi}. The matrix governing the compatibility geometry is therefore
\[
H_\xi
=
\frac{t^{4}}{2}
\begin{pmatrix}
(1+s^2)^2 & (1-s^2)^2 & 4s^2\cos^2\vartheta & 4s^2\sin^2\vartheta \\
(1-s^2)^2 & (1+s^2)^2 & 4s^2\sin^2\vartheta & 4s^2\cos^2\vartheta \\
4s^2\cos^2\vartheta & 4s^2\sin^2\vartheta & (1+s^2)^2 & (1-s^2)^2 \\
4s^2\sin^2\vartheta & 4s^2\cos^2\vartheta & (1-s^2)^2 & (1+s^2)^2
\end{pmatrix}\,,
\]
where $t=(1+s^2)^{-1/2}$. This is block-circulant with circulant blocks, and hence diagonal in the Hadamard basis \eqref{eq:hadamard-basis}. The eigenvalues are those in Eq.~\eqref{eq:HW-evalues}. The shift-covariant incoherent observables $\P_{\mathbf r}\in \mathscr I^{\rm sc}_{\!4}$ are determined by a seed distribution $\mathbf r\in\Delta_4$ through
\begin{equation}\label{app:eq:p-from-r-app}
p_{kl}(n,m)=r_{n\ominus k,\;m\ominus l},
\end{equation}
where $\ominus$ denotes subtraction modulo $2$.

\begin{corollary}\label{app:cor:HW-covariant-app}
Let $\xi\in\mathcal E_{>}(2)$ be the Heisenberg--Weyl extremal in Eq. \eqref{eq:hw-extremal-xi}. Then
\begin{equation}
\mathscr C_\xi[\mathscr I^{\rm sc}_{\!4}]
=
\Bigl\{
\P_{\mathbf r}\in\mathscr I^{\rm sc}_{\!4}
\;\Bigm|\;
\sum_{(n,m)\in\Omega_0}
\frac{\hat r_{nm}^{\,2}}{\lambda_{nm}}
\leq 1
\Bigr\}\,,
\end{equation}
where $\hat r_{nm}:=\sum_{k,l=0}^1(-1)^{kn+lm}r_{kl}$, $\Omega_0=\Omega\setminus\{(0,0)\}$, and $\lambda_{nm}$ are the three nontrivial eigenvalues of $H_\xi$ given in Eq.~\eqref{eq:HW-evalues}.
\end{corollary}

\begin{proof}
Let $(n,m)\in\Omega$. For the corresponding effect of $\P_{\mathbf r}$, the
coefficient vector is $p_{kl}(n,m)=r_{n\ominus k,\;m\ominus l}$. Its Fourier coefficients are therefore
\begin{align}
\hat p_{ab}(n,m)
&=
\sum_{k,l=0}^1 (-1)^{ka+lb}\,r_{n\ominus k,\;m\ominus l}\\
&=
\sum_{k',l'=0}^1 (-1)^{(n\ominus k')a+(m\ominus l')b}\,r_{k'l'}\\
&=
(-1)^{na+mb}\hat r_{ab},
\end{align}
where $k'=n\ominus k$ and $l'=m\ominus l$. In particular, $\hat p_{00}(n,m)=\hat r_{00}=1$, since $\mathbf r\in\Delta_4$. Hence, the vector $\mathbf x^\P_\xi(n,m)$ in Eq.~\eqref{eq:x-coordinates} has coordinates
\[
x_{ab}(n,m)
=
\frac{\hat p_{ab}(n,m)}{\hat p_{00}(n,m)}
=
(-1)^{na+mb}\hat r_{ab}\,,
\]
for every $(a,b)\neq(0,0)$. Thus $x_{ab}(n,m)^2=\hat r_{ab}^{\,2}$ is independent of the outcome
$(n,m)$. The claim now follows directly from Proposition~\ref{ellipses}.
\end{proof}

We now determine the threshold $\alpha^*(\xi)$, defined in Eq.~\eqref{eq:alpha-thresh}, at which $\P_\alpha$ enters $\mathscr C_\xi[\mathscr I^{\rm sc}_{\!4}]$.
\begin{corollary}\label{app:cor:dariano-app}
Let $\xi(s,\vartheta)\in\mathcal E_{>}(2)$ be a Heisenberg--Weyl extremal,
with eigenvalues $\lambda_{01}(s,\vartheta)$,
$\lambda_{10}(s,\vartheta)$ and $\lambda_{11}(s,\vartheta)$ of $H_\xi$ given
in Eq.~\eqref{eq:HW-evalues}. Then
\[
\P_\alpha\in\mathscr C_\xi[\mathscr I^{\rm sc}_{\!4}]
\qquad\Longleftrightarrow\qquad
\alpha\le \alpha^*(s,\vartheta),
\]
where
\[
\alpha^*(s,\vartheta)
=
\left(
\frac{1}{\lambda_{01}(s,\vartheta)}
+\frac{1}{\lambda_{10}(s,\vartheta)}
+\frac{1}{\lambda_{11}(s,\vartheta)}
\right)^{-\frac 12}.
\]
In particular,
\[
\alpha^*(s,\vartheta)=
\Biggl[
\frac{(1+s^{2})^{2}}{4s^{2}}
\left(
\sec^{2}\vartheta
+\csc^{2}\vartheta
+\frac{4s^{2}}{(1-s^{2})^{2}}
\right)
\Biggr]^{-\frac 12}.
\]
\end{corollary}

\begin{proof}
For $\mathbf r_\alpha=\frac14(1+3\alpha,1-\alpha,1-\alpha,1-\alpha)$, the nontrivial Fourier coefficients satisfy $\hat r_{01}=\hat r_{10}=\hat r_{11}=\alpha$. Thus Corollary~\ref{app:cor:HW-covariant-app} gives
\[
\alpha^2
\left(
\frac{1}{\lambda_{01}}
+\frac{1}{\lambda_{10}}
+\frac{1}{\lambda_{11}}
\right)
\le 1,
\]
which gives the first formula for $\alpha^*(s,\vartheta)$. Substituting the
eigenvalues from Eq.~\eqref{eq:HW-evalues} yields the second formula.
\end{proof}

\subsection{Semi-SIC extremals}
\label{app:semi-sic-proof}

We recall the semi-SIC family used in Sec.~\ref{sec:semi-sic-extremals}. For
$r=2$, the POVM effects are
\[
\M(n,m)=e^\xi_{nm}\,|a_{nm}\rangle\langle a_{nm}|,
\qquad n,m=0,1,
\]
where
\begin{align*}
\ket{a_{00}} &= \ket{0},\\
\ket{a_{01}} &= \gamma\ket{0}+\sqrt{1-\gamma^2}\ket{1},\\
\ket{a_{10}} &= \frac{1}{\sqrt{3}}\bigl(\ket{0}-\sqrt{2}\,e^{i\theta}\ket{1}\bigr),\\
\ket{a_{11}} &= \frac{1}{\sqrt{3}}\bigl(\ket{0}-\sqrt{2}\,e^{-i\theta}\ket{1}\bigr),
\end{align*}
with
\[
\gamma=\frac{2\sqrt{\beta}}{1-\sqrt{1-12\beta}},
\,\,\,
\theta=\cos^{-1}\!\left(
\frac{\sqrt{1-8\beta-\sqrt{1-12\beta}}}{4\sqrt{\beta}}
\right),
\]
and $\beta\in(1/16,1/12]$. The vector $\mathbf e^\xi=(e^\xi_{00},e^\xi_{01},e^\xi_{10},e^\xi_{11})^T$ is given Eq. \eqref{eq:semi-sic-e} of the main text. The corresponding coherence matrix can be found in Eq. \eqref{eq:semi-sic-xi}. The relevant matrix governing the compatibility region is
\begin{equation}
H_\xi = \begin{pmatrix}
e_+ & \frac{\beta}{e_+} & \frac{\beta}{\sqrt{e_+e_-}} & \frac{\beta}{\sqrt{e_+e_-}}\\
\frac{\beta}{e_+}& e_+ & \frac{\beta}{\sqrt{e_+e_-}}&\frac{\beta}{\sqrt{e_+e_-}}\\
\frac{\beta}{\sqrt{e_+e_-}}& \frac{\beta}{\sqrt{e_+e_-}}&  e_- & \frac{\beta}{e_-}\\
\frac{\beta}{\sqrt{e_+e_-}}& \frac{\beta}{\sqrt{e_+e_-}}& \frac{\beta}{e_-} & e_-
\end{pmatrix},
\end{equation}
where we have suppressed $\xi$-dependence in $e_\pm$ for notational simplicity. Equivalently, $H_\xi=A+|u\rangle\langle u|$, where $A=a_+P_+ + a_-P_-$, and
\[
P_+=|0\rangle\langle0|+|1\rangle\langle1|,
\qquad
P_-=|2\rangle\langle2|+|3\rangle\langle3|,
\]
with $a_\pm=e_\pm-\frac{\beta}{e_\pm}$, and
\[
|u\rangle
=
\sqrt{\frac{\beta}{e_+}}(|0\rangle+|1\rangle)
+
\sqrt{\frac{\beta}{e_-}}(|2\rangle+|3\rangle).
\]
The eigenbasis is given in Eq. \eqref{eq:semisic-basis}. The vector ${\bf v}_{00}=\sqrt{\mathbf e^\xi}/\sqrt2$ has eigenvalue $1$, and $\mathbf v_{01}$ is an orthogonal eigenvector in the symmetric subspace with eigenvalue $1/3$. The antisymmetric vectors $\mathbf v_{10}$ and $\mathbf v_{11}$ lie in $u^\perp$, and hence have eigenvalues $a_+$ and $a_-$, respectively.
Using $e_+e_-=3\beta$, these simplify to
\[
a_\pm=\frac{1\pm 2\sqrt{1-12\beta}}{3}.
\]
From Eq.~\eqref{eq:v-to-u}, the corresponding coordinate vectors are
\[
\mathbf u_{01}:=\sqrt{\frac{3\beta}{2}}(1,1,-1,-1)^T,
\]
\[
\mathbf u_{10}:=\sqrt{\frac{e_+}{2}}(1,-1,0,0)^T,
\qquad
\mathbf u_{11}:=\sqrt{\frac{e_-}{2}}(0,0,1,-1)^T.
\]
This gives the coordinates in Eq.~\eqref{eq:semi-sic-x}.

We now give the covariant compatibility region stated in Sec.~\ref{sec:semi-sic-extremals}.

\begin{corollary}\label{app:cor:semi-sics-app}
Let $\xi\in\mathcal E_{>}(2)$ be the semi-SIC extremal in Eq. \eqref{eq:semi-sic-xi}. Then
\begin{equation}\label{app:eq:semi-sic-hw}
\mathscr C_\xi[\mathscr I^{\rm sc}_{\!4}]
=
\left\{
\P_{\mathbf r}\in\mathscr I^{\rm sc}_{\!4}
\;\middle|\;
\hat{\mathbf r}\in\mathscr E^+_\xi\cap\mathscr E^-_\xi
\right\}\,,
\end{equation}
where $\hat{\mathbf r}:=(\hat r_{01},\hat r_{10},\hat r_{11})$ is the nontrivial Fourier-coordinate vector associated with
$\mathbf r\in\Delta_4$, and $\mathscr E^\pm_\xi$ are the ellipsoids
\begin{equation}\label{app:eq:Epm-semisic}
\mathscr E^\pm_\xi
:=
\left\{
\hat{\mathbf r}\in\mathbb R^3 \;\middle|\;
\mathscr Q^\pm_\xi(\hat{\mathbf r})
\le
\frac14+\frac{\delta^2}{16\kappa}
\right\},
\end{equation}
where $\delta:=e_+-e_-$, $\kappa:=9\beta-\frac{\delta^2}{4}$, and
\[
\begin{aligned}
\mathscr Q^\pm_\xi(\hat{\mathbf r})
:={}&\,
\frac{e_+}{4a_+}(\hat r_{01}\pm \hat r_{11})^2
+\frac{e_-}{4a_-}(\hat r_{01}\mp \hat r_{11})^2  \\
&+
\kappa
\left(
\hat r_{10}\mp\frac{\delta}{4\kappa}
\right)^2 .
\end{aligned}
\]
\end{corollary}

\begin{proof}
Let $\P_{\mathbf r}$ be shift-covariant, with seed distribution $\mathbf r\in\Delta_4$, and let $\mathbf p(n,m)$ denote the coefficient vector of the effect corresponding to $(n,m)\in\Omega$. As shown previously, we have $\hat p_{ab}(n,m)=(-1)^{na+mb}\hat r_{ab}$. Using
\[
p_{00}-p_{01}=\frac{\hat p_{01}+\hat p_{11}}{2},\qquad
p_{10}-p_{11}=\frac{\hat p_{01}-\hat p_{11}}{2},
\]
and $p_{00}+p_{01}-p_{10}-p_{11}=\hat p_{10}$, we obtain
\begin{align*}
\mathbf u_1^T\mathbf p(n,m)
&=
\sqrt{\frac{3\beta}{2}}\,(-1)^n\hat r_{10},\\
\mathbf u_2^T\mathbf p(n,m)
&=
\sqrt{\frac{e_+}{2}}\,
\frac{(-1)^m}{2}\bigl(\hat r_{01}+(-1)^n\hat r_{11}\bigr),\\
\mathbf u_3^T\mathbf p(n,m)
&=
\sqrt{\frac{e_-}{2}}\,
\frac{(-1)^m}{2}\bigl(\hat r_{01}-(-1)^n\hat r_{11}\bigr).
\end{align*}
Moreover, $[\mathbf e^\xi]^T\mathbf p(n,m)=\frac12(1+(-1)^n(e_+-e_-)\hat r_{10})$. Hence the squared coordinates in Eq.~\eqref{eq:x-coordinates} depend only on the value of $n$, so the four outcomes split into two inequivalent classes, corresponding to $n=0$ and $n=1$. Writing $\delta:=e_+-e_-$, and $d_\pm:=\frac12(1\pm\delta\hat r_{10})$, with the upper sign corresponding to $n=0$ and the lower sign to $n=1$, the vector $\textbf x^\P=(x_{01},x_{10},x_{11})$ in Eq. \eqref{eq:x-coordinates} becomes
\[
\begin{aligned}
x_{01}
&=
\sqrt{2}\,\frac{\mathbf u_1^T\mathbf p}{[\mathbf e^\xi]^T\mathbf p}
=
\frac{\sqrt{3\beta}\,\hat r_{10}}{d_\pm},
\\[0.8em]
x_{10}
&=
\sqrt{2}\,\frac{\mathbf u_2^T\mathbf p}{[\mathbf e^\xi]^T\mathbf p}
=
\frac{\sqrt{e_+}}{2d_\pm}\,
(\hat r_{01}\pm\hat r_{11}),
\\[0.8em]
x_{11}
&=
\sqrt{2}\,\frac{\mathbf u_3^T\mathbf p}{[\mathbf e^\xi]^T\mathbf p}
=
\frac{\sqrt{e_-}}{2d_\pm}\,
(\hat r_{01}\mp\hat r_{11}).
\end{aligned}
\]
Substituting these expressions into the ellipsoid condition $3x_{01}^2+x_{10}^2/a_++x_{11}^2/a_-\le 1$ gives
\[
\frac{e_+}{4a_+}(\hat r_{01}\pm\hat r_{11})^2
+\frac{e_-}{4a_-}(\hat r_{01}\mp\hat r_{11})^2
+9\beta\hat r_{10}^2
\le
\frac14(1\pm\delta\hat r_{10})^2\,,
\]
which can be rewritten as
\[
\begin{aligned}
\frac{e_+}{4a_+}(\hat r_{01}\pm\hat r_{11})^2
&+\frac{e_-}{4a_-}(\hat r_{01}\mp\hat r_{11})^2 \\
&+\left(9\beta-\frac{\delta^2}{4}\right)\hat r_{10}^2
\mp \frac{\delta}{2}\hat r_{10}
\le \frac14 .
\end{aligned}
\]
Setting $\kappa:=9\beta-\frac{\delta^2}{4}$ and completing the square in $\hat r_{10}$ yields
\[
\begin{aligned}
\frac{e_+}{4a_+}(\hat r_{01}\pm\hat r_{11})^2
&+\frac{e_-}{4a_-}(\hat r_{01}\mp\hat r_{11})^2 \\
&+\kappa
\left(
\hat r_{10}\mp\frac{\delta}{4\kappa}
\right)^2
\le
\frac14+\frac{\delta^2}{16\kappa}\,,
\end{aligned}
\]
which is the required condition. Therefore $\P_{\mathbf r}\in\mathscr C_\xi[\mathscr I^{\rm sc}_{\!4}]$ if and only if $\hat{\mathbf r}\in \mathscr E^+_\xi\cap \mathscr E^-_\xi $.
\end{proof}

\subsection{Non-MIC extremals}
\label{app:nonmic-family}

Here we give the details for the non-MIC extremal family considered in Sec.~\ref{sec:nonmic-family}, with coherence matrix in Eq. \eqref{eq:non-mic-xi}. This family is obtained from the example in \cite{buscemi05}, corresponding to $\phi=\pi/2$, by introducing a relative phase $e^{i\phi}$ in one of the structure vectors. The structure vectors may be chosen as
\[
\begin{aligned}
|a_{00}\rangle&=|0\rangle, & |a_{01}\rangle&=|1\rangle,\\
|a_{10}\rangle&=\tfrac{1}{\sqrt2}(|0\rangle+|1\rangle), &
|a_{11}\rangle&=\tfrac{1}{\sqrt2}(|0\rangle+e^{i\phi}|1\rangle),
\end{aligned}
\]
which shows that the rank is $r=2$. The Gram matrix $G_\xi$ that determines the compatibility geometry has the simple form
\begin{equation}\label{eq:gram-non-mic}
G_\xi=
\left(
\begin{array}{cccc}
1 & 0 & \frac12 & \frac12\\[0.2em]
0 & 1 & \frac12 & \frac12\\[0.2em]
\frac12 & \frac12 & 1 & c\\[0.2em]
\frac12 & \frac12 & c & 1
\end{array}
\right)\,,
\end{equation}
where $c:=\cos^2(\phi/2)$ and $\phi\in(0,\pi)$. Hence $c\in(0,1)$. We now see that the inverse of $G_\xi$ exists for $c\in(0,1)$:
\[
G_\xi^{-1}
=
\frac{1}{2c}
\begin{pmatrix}
2c+1 & 1 & -1 & -1 \\[0.3em]
1 & 2c+1 & -1 & -1 \\[0.3em]
-1 & -1 &
\dfrac{1}{1-c} &
\dfrac{1-2c}{1-c} \\[0.7em]
-1 & -1 &
\dfrac{1-2c}{1-c} &
\dfrac{1}{1-c}
\end{pmatrix},
\]
which shows that $\xi\in \mathcal E(2)$. However, since $|a_{00}\rangle\langle a_{00}|+|a_{01}\rangle\langle a_{01}|=\id$, one has $\mathbf e^\xi=(1,1,0,0)$, and therefore $\xi\in\mathcal E_{\geq}(2)$ but $\xi\notin\mathcal E_{>}(2)$, as already noted in the main text, so these extremals are not MICs. Accordingly, the matrix $H_\xi$ cannot be used for determining the compatibility region---in fact, we see that it is trivial: $H_\xi=\id\oplus \mathbf 0$.

Instead, from Eq.~\eqref{eq:K-matrix}, we find the relevant matrix:
\begin{equation}\label{app:eq:K-nonmic}
K_{\xi}
=
\frac{1}{2c}
\begin{pmatrix}
-1 & 2c-1 & 1 & 1 \\[0.3em]
2c-1 & -1 & 1 & 1 \\[0.3em]
1 & 1 &
-\dfrac{1}{1-c} &
\dfrac{2c-1}{1-c} \\[0.7em]
1 & 1 &
\dfrac{2c-1}{1-c} &
-\dfrac{1}{1-c}
\end{pmatrix}.
\end{equation}

Fixing an outcome $j$, let $\mathbf p(j)=(p_{00},p_{01},p_{10},p_{11})^T$, and set $h:=\frac{1}{2}(p_{00}+p_{01})$. A direct calculation gives
\begin{align*}
\mathbf p(j)^T K_{\xi}\,\mathbf p(j)
&=
-\frac{1}{c}\Bigl(
(p_{10}-h)^2+(p_{11}-h)^2 \nonumber\\
&
-(2c-1)(p_{10}-h)(p_{11}-h)-c\,p_{00}p_{01}
\Bigr),
\end{align*}
and the compatibility condition $\mathbf p(j)^TK_\xi\mathbf p(j)\ge0$ gives Eq.~\eqref{eq:non-mic-criterion}.

We now derive Eq.~\eqref{eq:non-mic-ellipsoid} for the covariant seed distribution $\mathbf r$. Writing $\mathbf r$ in the Hadamard basis \eqref{eq:hadamard-basis}, one has
\[
\mathbf r
=
\frac12\mathbf v_{00}
+\frac{\hat r_{01}}{2}\mathbf v_{01}
+\frac{\hat r_{10}}{2}\mathbf v_{10}
+\frac{\hat r_{11}}{2}\mathbf v_{11},
\]
where $\hat r_{nm}$ are the Fourier coefficients of $\mathbf r$. We stress here that these vectors are now \emph{not} eigenvectors of $H_\xi$; they are the eigenvectors in Eq. \eqref{eq:hadamard-basis} of $H_\xi$ for the HW-extremals. 

Using Eq.~\eqref{app:eq:K-nonmic}, one finds
\[
\begin{aligned}
K_{\xi}\mathbf v_{00}
&=\frac12\mathbf v_{00}+\frac12\mathbf v_{10},\\
K_{\xi}\mathbf v_{01}
&=-\frac{2-c}{2(1-c)}\mathbf v_{01}
+\frac{c}{2(1-c)}\mathbf v_{11},\\
K_{\xi}\mathbf v_{10}
&=\frac12\mathbf v_{00}-\frac{4-c}{2c}\mathbf v_{10},\\
K_{\xi}\mathbf v_{11}
&=\frac{c}{2(1-c)}\mathbf v_{01}
-\frac{2-c}{2(1-c)}\mathbf v_{11}.
\end{aligned}
\]
Hence
\[
\begin{aligned}
K_{\xi}\,\mathbf r
={}&\frac14(1+\hat r_{10})\mathbf v_{00}
+\frac14\left(1-\frac{4-c}{c}\hat r_{10}\right)\mathbf v_{10}\\
&+\frac{1}{4(1-c)}\bigl(-(2-c)\hat r_{01}+c\,\hat r_{11}\bigr)\mathbf v_{01}\\
&+\frac{1}{4(1-c)}\bigl(c\,\hat r_{01}-(2-c)\hat r_{11}\bigr)\mathbf v_{11}.
\end{aligned}
\]
Taking the inner product with $\mathbf r$ gives
\begin{equation*}
\begin{aligned}
\mathbf r^T K_{\xi}\,\mathbf r
={}&\frac{1}{8}
+\frac{\hat r_{10}}{4}
-\frac{4-c}{8c}\hat r_{10}^2\\
&-\frac{2-c}{8(1-c)}(\hat r_{01}^2+\hat r_{11}^2)
+\frac{c}{4(1-c)}\hat r_{01}\hat r_{11}.
\end{aligned}
\end{equation*}
Thus the compatibility condition is equivalent to
\begin{equation}\label{eq:app-1-r}
\frac{4-c}{8c}\hat r_{10}^2-\frac{\hat r_{10}}{4}
+\frac{2-c}{8(1-c)}(\hat r_{01}^2+\hat r_{11}^2)
-\frac{c}{4(1-c)}\hat r_{01}\hat r_{11}
\le
\frac{1}{8}.
\end{equation}
Now observe that
\[
\begin{aligned}
(\hat r_{01}+\hat r_{11})^2+\frac{1}{1-c}(\hat r_{01}-\hat r_{11})^2
&=
\frac{2-c}{1-c}(\hat r_{01}^2+\hat r_{11}^2) \\
&\quad -\frac{2c}{1-c}\hat r_{01}\hat r_{11}.
\end{aligned}
\]
Hence Eq.~\eqref{eq:app-1-r} can be rewritten as
\begin{equation}\label{eq:app-2-r}
(\hat r_{01}+\hat r_{11})^2+\frac{1}{1-c}(\hat r_{01}-\hat r_{11})^2
+\frac{4-c}{c}\hat r_{10}^2-2\hat r_{10}
\le
1.
\end{equation}
Completing the square in $\hat r_{10}$,
\[
\frac{4-c}{c}\hat r_{10}^2-2\hat r_{10}
=
\frac{4-c}{c}\left(\hat r_{10}-\frac{c}{4-c}\right)^2
-\frac{c}{4-c},
\]
and substituting this into Eq.~\eqref{eq:app-2-r} yields Eq.~\eqref{eq:non-mic-ellipsoid}:
\[
\begin{aligned}
(\hat r_{01}+\hat r_{11})^2+\frac{1}{1-c}(\hat r_{01}-\hat r_{11})^2
&+\frac{4-c}{c}\left(\hat r_{10}-\frac{c}{4-c}\right)^2 \\
&\le \frac{4}{4-c}\,.
\end{aligned}
\]

We can now give the covariant compatibility region as follows.

\begin{corollary}\label{app:cor:nonmic-hw-app}
Let $\xi\in\mathcal E_{\geq}(2)$ be the non-MIC extremal in Eq.~\eqref{eq:non-mic-xi}. Then
\begin{equation*}
\mathscr C_\xi[\mathscr I^{\rm sc}_{\!4}]
=
\left\{
\P_{\mathbf r}\in\mathscr I^{\rm sc}_{\!4}
\;\middle|\;
\hat{\mathbf r}\in\mathscr E^+_\xi\cap\mathscr E^-_\xi
\right\}\,,
\end{equation*}
where $\hat{\mathbf r}:=(\hat r_{01},\hat r_{10},\hat r_{11})$ is the nontrivial Fourier-coordinate vector associated with $\mathbf r\in\Delta_4$, and
$\mathscr E^\pm_\xi$ are the ellipsoids
\begin{equation}\label{app:eq:Epm-nonmic}
\mathscr E^\pm_\xi :=
\left\{
\hat{\mathbf r}\in\mathbb R^3 \;\middle|\;
\mathscr Q^\pm_\xi(\hat{\mathbf r})\le \frac{4c}{4-c}
\right\},
\end{equation}
with
\[
\begin{aligned}
\mathscr Q^\pm_\xi(\hat{\mathbf r})
:={}&\,
c(\hat r_{01}\pm \hat r_{11})^2
+\frac{c}{1-c}(\hat r_{01}\mp \hat r_{11})^2\\
&+(4-c)\left(\hat r_{10}\mp\frac{c}{4-c}\right)^2 .
\end{aligned}
\]
\end{corollary}

\begin{proof}
Under translation by $(n,m)$, the Fourier coefficients of the seed distribution transform as
\[
(\hat r_{01},\hat r_{10},\hat r_{11})\mapsto
\bigl((-1)^m \hat r_{01},\;(-1)^n \hat r_{10},\;(-1)^{n+m}\hat r_{11}\bigr).
\]
Applying Eq.~\eqref{eq:non-mic-ellipsoid} to the four translated effects gives two inequivalent conditions, corresponding to the outcome classes $(0,0)\sim(0,1)$ and $(1,0)\sim(1,1)$. These two conditions are
\[
\mathcal Q^+_\xi(\hat{\mathbf r})\le \frac{4c}{4-c},
\qquad
\mathcal Q^-_\xi(\hat{\mathbf r})\le \frac{4c}{4-c}.
\]
Hence $\P_{\mathbf r}\in\mathscr C_\xi[\mathscr I^{\rm sc}_{\!4}]$ if and only if
$\hat{\mathbf r}\in \mathscr E^+_\xi\cap \mathscr E^-_\xi$.
\end{proof}

\end{document}